\documentclass[a4paper,UKenglish]{lipics}

\usepackage{microtype}

\usepackage{amsmath}
\usepackage{amssymb}
\usepackage{graphicx}
\usepackage{algorithm}
\usepackage{algpseudocode}
\usepackage{color}
\usepackage{hyperref}
\usepackage{gensymb}

\renewcommand{\algorithmicrequire}{\textbf{Initialization:}}

\newcommand{\removed}[1]{}

\newcommand{\rem}{\mathbf}

\newcommand{\bbbn}{\mathbb{N}}

\newcommand{\ra}{\rightarrow}
\newcommand{\rsa}{\rightsquigarrow}

\definecolor{DarkRed}{RGB}{182,11,1}

\newtheorem{observation}{Observation}
\theoremstyle{theorem}
\newtheorem{proposition}{Proposition}

\title{On the Transformation Capability of Feasible Mechanisms for Programmable Matter}
\titlerunning{On the Transformation Capability of Feasible Mechanisms for Programmable Matter} 

\author[1]{Othon Michail}
\author[2]{George Skretas}
\author[3]{Paul G. Spirakis}
\affil[1]{Department of Computer Science, University of Liverpool, UK\\
  \texttt{Othon.Michail@liverpool.ac.uk}}
\affil[2]{Computer Engineering and Informatics Department (CEID), University of Patras, Greece\\
  \texttt{skretas@ceid.upatras.gr}}
\affil[3]{Department of Computer Science, University of Liverpool, UK\\
  Computer Technology Institute and Press ``Diophantus'' (CTI), Patras, Greece\\
  \texttt{P.Spirakis@liverpool.ac.uk}}
\authorrunning{O. Michail, G. Skretas, and P.\,G. Spirakis} 

\Copyright{Othon Michail, George Skretas, and Paul G. Spirakis}


\begin{document}

\maketitle

\begin{abstract}
In this work, we study theoretical models of \emph{programmable matter} systems. The systems under consideration consist of spherical modules, kept together by magnetic forces and able to perform two minimal mechanical operations (or movements): \emph{rotate} around a neighbor and \emph{slide} over a line. In terms of modeling, there are $n$ nodes arranged in a 2-dimensional grid and forming some initial \emph{shape}. The goal is for the initial shape $A$ to \emph{transform} to some target shape $B$ by a sequence of movements. Most of the paper focuses on \emph{transformability} questions, meaning whether it is in principle feasible to transform a given shape to another. We first consider the case in which only rotation is available to the nodes. Our main result is that deciding whether two given shapes $A$ and $B$ can be transformed to each other, is in $\rem{P}$. We then insist on rotation only and impose the restriction that the nodes must maintain global connectivity throughout the transformation. We prove that the corresponding transformability question is in $\rem{PSPACE}$ and study the problem of determining the minimum \emph{seeds} that can make feasible, otherwise infeasible transformations. Next we allow both rotations and slidings and prove universality: any two connected shapes $A,B$ of the same order, can be transformed to each other without breaking connectivity. The worst-case number of movements of the generic strategy is $\Omega(n^2)$. We improve this to $O(n)$ parallel time, by a pipelining strategy, and prove optimality of both by matching lower bounds. In the last part of the paper, we turn our attention to distributed transformations. The nodes are now distributed processes able to perform communicate-compute-move rounds. We provide distributed algorithms for a general type of transformations.
\end{abstract}

\section{Introduction}
\label{sec:intro}

\emph{Programmable matter} refers to any type of matter that can \emph{algorithmically} change its physical properties. ``Algorithmically'' means that the change (or \emph{transformation}) is the result of executing an \emph{underlying program}. Depending on the implementation, the program could either be a \emph{centralized algorithm} capable of controlling the whole programmable matter system (\emph{external} control) or a \emph{decentralized protocol} stored in the material itself and executed by various sub-components of the system (\emph{internal} control). For a concrete example, imagine a material formed by a collection of spherical nanomodules kept together by magnetic forces. Each module is capable of storing (in some internal representation) and executing a simple program that handles communication with nearby modules and that controls the module's electromagnets, in a way that allows the module to \emph{rotate} or \emph{slide} over neighboring modules. Such a material would be able to adjust its \emph{shape} in a programmable way. Other examples of physical properties of interest for real applications would be connectivity, color \cite{LYS01,CLS11}, and strength of the material.

Computer scientists, nanoscientists, and engineers are more and more joining their forces towards the development of such programmable materials and have already produced some first impressive outcomes (even though it is evident that there is much more work to be done in the direction of real systems), such as programmed DNA molecules that self-assemble into desired structures \cite{Ro06,Do12} and large collectives of tiny identical robots that orchestrate resembling a single multi-robot organism (Kilobot system) \cite{RCN14}. Other systems for programmable matter include the Robot Pebbles \cite{GKR10}, consisting of 1cm cubic programmable matter modules able to form 2-dimensional (usually abbreviated ``2D'') shapes through self-disassembly, and the Millimotein \cite{KCL12}, a chain of programmable matter which can fold itself into digitized approximations of arbitrary 3-dimensional (usually abbreviated ``3D'') shapes. Ambitious long-term applications of programmable materials include molecular computers, collectives of nanorobots injected into the human circulatory system for monitoring and treating diseases, or even self-reproducing and self-healing machines.

Apart from the fact that systems work is still in its infancy, there is also an apparent lack of unifying formalism and theoretical treatment. The following are some of the very few exceptions aiming at understanding the fundamental possibilities and limitations of this prospective. The area of \emph{algorithmic self-assembly} tries to understand how to program molecules (mainly DNA strands) to manipulate themselves, grow into machines and at the same time control their own growth \cite{Do12}. The theoretical model guiding the study in algorithmic self-assembly is the Abstract Tile Assembly Model (aTAM) \cite{Wi98,RW00} and variations. Recently, a model, called the \emph{nubot} model, was proposed for studying the complexity of self-assembled structures with active molecular components \cite{WCG13}. This model ``is inspired by biology's fantastic ability to assemble biomolecules that form systems with complicated structure and dynamics, from molecular motors that walk on rigid tracks and proteins that dynamically alter the structure of the cell during mitosis, to embryonic development where large-scale complicated organisms efficiently grow from a single cell'' \cite{WCG13}. Another very recent model, called the \emph{Network Constructors} model, studied what stable networks can be constructed by a population of finite-automata that interact randomly like molecules in a well-mixed solution and can establish bonds with each other according to the rules of a common small protocol \cite{MS16a}. The development of Network Constructors was based on the \emph{Population Protocol} model of Angluin \emph{et al.} \cite{AADFP06}, that does not include the capability of creating bonds and focuses more on the computation of functions on inputs. A very interesting fact about population protocols is that they are formally equivalent to \emph{chemical reaction networks} (CRNs), ``which model chemistry in a \emph{well-mixed solution} and are widely used to describe information processing occurring in natural cellular regulatory networks'' \cite{Do14}. Also the recently proposed \emph{Amoebot} model, ``offers a versatile framework to model self-organizing particles and facilitates rigorous algorithmic research in the area of programmable matter'' \cite{DDGR14,DGR15,DGPR16}. An indication of the potential that the research community sees in this effort, is the 1st Dagstuhl Seminar on ``Algorithmic Foundations of Programmable Matter'', which took place in June 2016 and attracted leading scientist (both theoreticians and practitioners) from Algorithms, Distributed Computing, Robotics, and DNA Self-Assembly, with the aim at joining their forces to push forward this emerging subject.

Each theoretical approach, and to be more precise, each individual model, has its own beauty and has lead to different insights and developments regarding potential programmable matter systems of the future and in some cases to very intriguing technical problems and open questions. Still, it seems that the right way for theory to boost the development of more refined real systems is to reveal the \emph{transformation capabilities of mechanisms and technologies that are available now}, rather than by exploring the unlimited variety of theoretical models that are not expected to correspond to a real implementation in the near future.

In this paper, we follow such an approach, by studying the transformation capabilities of models for programmable matter, which are based on minimal mechanical capabilities, easily implementable by existing technology.

\subsection{Our Approach}

We study a minimal programmable matter system consisting of $n$ cycle-shaped modules, with each module (or \emph{node}) occupying at any given time a cell of the 2D grid (no two nodes can occupy the same cell at the same time). Therefore, the composition of the programmable matter systems under consideration is discrete. Our main question throughout is whether an initial arrangement of the material can transform (either in principle, e.g., by an external authority, or by itself) to some other target arrangement. In more technical terms, we are provided with an \emph{initial shape} $A$ and a \emph{target shape} $B$ and we are asked whether $A$ \emph{can be transformed to} $B$ via a sequence of \emph{valid} transformation steps. Usually, a step consists either of a \emph{valid movement} of a single node (in the \emph{sequential case}) or of more than one nodes at the same time (in the \emph{parallel case}). We consider two quite primitive types of movement. The first one, called \emph{rotation}, allows a node to rotate 90\degree\ around one of its neighbors either clockwise or counterclockwise (see, e.g., Figure \ref{fig:line-folding} in Section \ref{sec:rotation}) and the second one, called \emph{sliding}, allows a node to slide by one position ``over'' two neighboring nodes (see, e.g., Figure \ref{fig:sliding-definition} in Section \ref{sec:rotation-sliding}). Both movements succeed only if the whole direction of movement is free of obstacles (i.e., other nodes blocking the way). More formal definitions are provided in Section \ref{sec:prel}. One part of the paper focuses on the case in which only rotation is available to the nodes and the other part studies the case in which both rotation and sliding are available. The latter case has been studied to some extent in the past in the, so called, \emph{metamorphic systems} \cite{DSY04a,DSY04b,DP04}, which makes those studies the closest to our approach.

For rotation only, we introduce the notion of \emph{color-consistence} and prove that if two shapes are not color-consistent then they cannot be transformed to each other. On the other hand color-consistence does not guarantee transformability as there is an infinite set of pairs $(A,B)$ such that $A$ and $B$ are color consistent but still they cannot be transformed to each other. At this point, observe that if $A$ can be transformed to $B$ then the inverse is also true, as all movements considered in this paper are \emph{reversible}. We distinguish two main types of transformations: those that are allowed to break the connectivity of the shape during the transformation and those that are not and call the corresponding problems {\sc Rot-Transformability} and {\sc RotC-Transformability}. We prove that {\sc RotC-Transformability} is a proper subset of {\sc Rot-Transformability} by showing that a line-folding problem is in {\sc Rot-Transformability}$\setminus${\sc RotC-Transformability}. Our main result regarding {\sc Rot-Transformability} is that {\sc Rot-Transformability} $\in \rem{P}$. To prove polynomial-time decidability, we prove that two shapes $A$ and $B$ are transformable to each other iff both $A$ and $B$ have at least one movement available (without any movement available, a shape is \emph{blocked} and can only trivially transform to itself). Therefore, transformability reduces to checking the availability of a movement in the initial and target shapes. The idea is that if a movement is available in a shape $A$, then there is always a way to extract from $A$ a \emph{2-line} (i.e., two neighboring nodes). Such a 2-line can move freely in any direction and can also extract further nodes to form a \emph{4-line}. A 4-line in turn can also move freely to any direction and is also capable of extracting nodes from the shape and transferring them, one at a time, to any desired target position. In this manner, the 4-line can transform $A$ to a line with leaves around it that is color-consistent to $A$ (based on a proposition that we prove, stating that any shape has a corresponding color-consistent line-with-leaves). Similarly, $B$, given that it is color-consistent with $A$, can be transformed by the same approach to exactly the same line-with-leaves, and then, by reversibility, it follows that $A$ and $B$ can be transformed to each other by using the line-with-leaves as an intermediate. This set of transformations do not guarantee the preservation of connectivity during the transformation. That is, even though the initial and target shapes considered are connected shapes, the shapes formed at intermediate steps of the transformation may very well be disconnected shapes.

We next study {\sc RotC-Transformability}, in which again the only available movement is rotation, but now connectivity of the material has to be preserved throughout the transformation. The property of preserving the connectivity is expected to be a crucial property for programmable matter systems, as it allows the material to maintain coherence and strength, to eliminate the need for wireless communication, and, finally, enables the development of more effective power supply schemes, in which the modules can share resources or in which the modules have no batteries but are instead constantly supplied with energy by a centralized source (or by a supernode that is part of the material itself). Such benefits can lead to simplified designs and potentially to reduced size of individual modules. We first prove that {\sc RotC-Transformability} $\in \rem{PSPACE}$. The rest of our results here are strongly based on the notion of a \emph{seed}. This stems from the observation that a large set of infeasible transformations become feasible by introducing to the initial shape an additional, and usually quite small, seed; i.e., a small shape that is being attached to some point of the initial shape. In particular, we prove that a \emph{3-line seed}, if placed appropriately, is sufficient to achieve folding of a line (otherwise impossible). We then investigate seeds that could serve as components capable of traveling the perimeter of an arbitrary connected shape $A$. Such shapes are very convenient as they are capable of ``simulating'' the \emph{universal transformation} techniques that are possible if we have both rotation and sliding movements available (discussed in the sequel). To this end, we prove that all seeds of size $\leq 4$ cannot serve for this purpose, by proving that they cannot even walk the perimeter of a simple line shape. Then we focus on a \emph{6-seed} and prove that such a seed is capable of walking the perimeter of a large family of shapes, called \emph{discrete-convex} shapes. This is a first indication, that there might be a large family of shapes that can be transformed to each other with rotation only and without breaking connectivity, by extracting a 6-seed and then exploiting to transfer nodes to the desired positions. To further support this, we prove that the 6-seed is capable of performing such transfers, by detaching pairs of nodes from the shape, attaching them to itself, thus forming an \emph{8-seed} and then being still capable to walk the perimeter of the shape.

Next, we consider the case in which both rotation and sliding are available and insist on connectivity preservation. We first provide a proof that this combination of simple movements is universal w.r.t. transformations, as any pair of connected shapes $A$ and $B$ of the same order, can be transformed to each other without ever breaking the connectivity throughout the transformation (a first proof of this fact had already appeared in \cite{DP04}). This generic transformation requires $\Theta(n^2)$ sequential movements in the worst case. By a potential-function argument we show that no transformation can improve on this worst-case complexity for some specific pairs of shapes and this lower bound is independent of connectivity preservation; it only depends on the inherent \emph{transformation-distance} between the shapes. To improve on this, either some sort of parallelism must be employed or more powerful movement mechanisms, e.g., movements of whole sub-shapes in one step. We investigate the former approach, and prove that there is a \emph{pipelining} general transformation strategy that improves the time to $O(n)$ (parallel time). We also give a matching $\Omega(n)$ lower bound. On the way, we also show that this parallel complexity is feasible even if the nodes are labeled, meaning that individual nodes must end up in specific positions of the target-shape.

Afterwards, we propose a distributed algorithm that transforms any compact shapes into a line using the rotation-sliding movement without breaking the connectivity of the shape. We note that a unique leader is required, each node has 4 ports and we aim to minimise the memory as much as possible. The communication is synchronous with each node broadcasting messages to its neighbours each turn. Following this, we propose an algorithm that transforms any shape into a line. We have the same requirements and communication and our goal again to minimize the amount of memory required in the system.

In Section \ref{subsec:further-related} we discuss further related literature. Section \ref{sec:prel} brings together all definitions and basic facts that are used throughout the paper. In Section \ref{sec:rotation}, we study programmable matter systems equipped only with rotation movement. In Section \ref{sec:rotation-connectivity}, we insist on rotation only, but additionally require from the material to maintain connectivity throughout the transformation. In Section \ref{sec:rotation-sliding}, we investigate the combined effect of rotation and sliding movements. Connectivity can always be preserved in this case. Section \ref{sec:distributed} focuses on distributed transformations having access to both rotation and sliding. Finally, in Section \ref{sec:conclusions} we conclude and give further research directions that are opened by our work.

\subsection{Further Related Work}
\label{subsec:further-related}

\noindent\textbf{Mobile and Reconfigurable Robotics.} There is a very rich literature on  mobile and reconfigurable robotics. In mobile (swarm) robotics systems and models, as are, for example, the models for robot gathering \cite{CFPS03,KKM10} and deployment \cite{SMO16} (cf., also \cite{FPS12}), geometric pattern formation \cite{SY99,DFSY15}, and connectivity preservation \cite{CKLL09}, the modules are usually robots equipped with some mobility mechanism making them free to move in any direction of the plane (and in some cases even continuously). In contrast, we only allow discrete movements relative to neighboring nodes. Modular self-reconfigurable robotic systems form an area on their own, focusing on aspects like the design, motion planning, and control of autonomous robotic modules \cite{BKR04,YSS07,ABD13,YUY16}. The model considered in this paper bears similarities to some of the models that have appeared in this area. The main difference is that we follow a more computation-theoretic approach, while the studies in this area usually follow a more applied perspective.\\

\noindent\textbf{Puzzles.} Puzzles are combinatorial one-player games, usually played on some sort of board. Typical questions of interest are whether a given puzzle is solvable and finding the solution with the fewest number of moves. Answers to these questions range from being in $\rem{P}$ up to $\rem{PSPACE}$-hard or even undecidable when some puzzles are generalized to the entire plane with unboundedly many pieces \cite{De01,HD05}. Famous examples of puzzles are the Fifteen Puzzle, Sliding Blocks, Rush Hour, Pushing Blocks, and Solitaire. Even though none of these is equivalent to the model considered here, the techniques that have been developed for solving and characterizing puzzles may turn very useful in the context of programmable matter systems. Actually, in some cases, such puzzles show up as special cases of the transformation problems considered here (e.g., the Fifteen Puzzle may be obtained if we restrict a transformation of node-labelled shapes to take place in a 4x4 square region).\\

\noindent\textbf{Passive Systems.} Most of the models discussed so far including the model under consideration in this paper, are \emph{active} models, meaning that the movements are in the complete control of the algorithm. In contrast, in \emph{passive} models the underlying algorithm cannot control the movements but in most cases it can decide in some way which movements to accept and which not. The typical assumption is that the movements are controlled by a \emph{scheduler} (possibly adversarial), which represents some dynamicity of the system or the environment. Population Protocols \cite{AADFP06, AAER07} and variants are a typical such example. For example, in Network Constructors \cite{MS16a} nodes move around randomly due to the dynamicity of the environment and when two of them interact the protocol can decide whether to establish a connection between them; that is, the protocol has some \emph{implicit} control of the system's dynamics. Another passive model, inspired from biological multicellular processes, was recently proposed by Emek and Uitto \cite{EU16}. Most models from the theory of algorithmic self-assembly, like the Abstract Tile Assembly Model (aTAM) \cite{Wi98,RW00}, fall also in this category. In this paper, we are only concerned with active systems. Hybrid models combining active capabilities and passive dynamics, remain an interesting open research direction.

\section{Preliminaries}
\label{sec:prel}

The programmable matter systems considered in this paper operate on a 2-dimensional square grid. As usual, each position (or \emph{cell}) of the grid is uniquely referred to by its $x$ and $y$ coordinates, where $x\geq 0$ corresponds to the row and $y\geq 0$ to the column. Such a system consists of a set $V$ of $n$ \emph{modules}, called \emph{nodes} throughout. Each node may be viewed as a spherical module fitting inside a cell of the grid. At any given time, each node $u\in V$ occupies a cell $o(u)=(o_x(u),o_y(u))=(i,j)$ (omitting the time index for simplicity here and also whenever clear from context) and no two nodes may occupy the same cell. In some cases, when a cell is occupied by a node we may refer to that cell by a color, e.g., \emph{black}, and when a cell is not occupied (i.e., it is empty) we usually refer to it as \emph{white}. At any given time $t$, the positioning of nodes on the grid defines an undirected \emph{neighboring relation} $E(t)\subset V\times V$, where $\{u,v\}\in E$ iff $o_x(u)=o_x(v)$ and $|o_y(u)-o_y(v)|=1$ or $o_y(u)=o_y(v)$ and $|o_x(u)-o_x(v)|=1$, that is, if $u$ and $v$ are either \emph{horizontal} or \emph{vertical} neighbors on the grid, respectively. It is immediate to observe that every node can have at most 4 neighbors at any given time. A more informative way to define the system at a given time $t$, and thus often more convenient, is as a mapping $P_t\colon\bbbn_{\geq 0}\times \bbbn_{\geq 0}\to \{0,1\}$ where $P_t(i,j)=1$ iff cell $(i,j)$ is occupied by a node.

At any given time $t$, $P_t^{-1}(1)$ defines a \emph{shape}. Such a shape is called \emph{connected} if $E(t)$ defines a connected graph. A connected shape is called \emph{convex} if for any two occupied cells, the line that connects their centers does not pass through an empty cell. We call a shape \emph{discrete-convex} if for any two occupied cells, belonging either to the same row or the same column, the line that connects their centers does not pass through an empty cell; i.e., in the latter we exclude diagonal lines.

In general, shapes can \emph{transform} to other shapes via a sequence of one or more \emph{movements} of individual nodes. Time consists of discrete \emph{steps} (or \emph{rounds}) and in every step, zero or more movements may occur, possibly following a computation sub-step either centralized or distributed, depending on the application. In the \emph{sequential} case, at most one movement may occur per step, and in the \emph{parallel} case any number of ``valid'' movements may occur in parallel. \footnote{By ``valid'', we mean here subject to the constraint that their whole movement paths correspond to pairwise disjoint sub-areas of the grid.} We consider two types of movements: (i) \emph{rotation} and (ii) \emph{sliding}. In both movements, a single node moves relative to one or more neighboring nodes as we explain now.

A single \emph{rotation} movement of a node $u$ is a 90\degree\ \emph{rotation} of $u$ around one of its neighbors. Let $(i,j)$ be the current position of $u$ and let its neighbor be $v$ occupying the cell $(i-1,j)$ (i.e., lying below $u$). Then $u$ can \emph{rotate} 90\degree\ clockwise (counterclockwise) around $v$ iff the cells $(i,j+1)$ and $(i-1,j+1)$ ($(i,j-1)$ and $(i-1,j-1)$, respectively) are both empty. By rotating the whole system 90\degree, 180\degree, and 270\degree , all possible rotation movements are defined analogously. See Figure \ref{fig:rotation-definition}.

\begin{figure}[!hbtp]
\centering{
\includegraphics[width=0.5\textwidth]{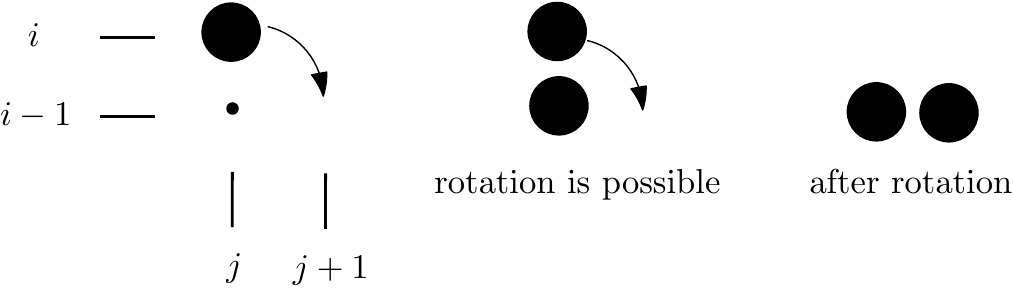}
}
\caption{Rotation to the right and down. A node on the black dot (in row $i-1$) and empty cells at positions $(i,j+1)$ and $(i-1,j+1)$ are required for this movement. Then an example movement is given.} \label{fig:rotation-definition}
\end{figure}

A single \emph{sliding} movement of a node $u$ is a one-step horizontal or vertical movement ``over'' a horizontal or vertical line of (neighboring) nodes of length 2. In particular, if $(i,j)$ is the current position of $u$, then $u$ can \emph{slide} rightwards to position $(i,j+1)$ iff $(i,j+1)$ is not occupied and there exist nodes at positions $(i-1,j)$ and $(i-1,j+1)$ or at positions $(i+1,j)$ and $(i+1,j+1)$, or both. Precisely the same definition holds for up, left, and down sliding movements by rotating the whole system 90\degree, 180\degree, and 270\degree\ counterclockwise, respectively. Intuitively, a node can slide one step in one direction, if there are two consecutive nodes either immediately ``below'' or immediately ``above'' that direction that can assist the node slide (see Figure \ref{fig:sliding-definition}). \footnote{Observe that there are plausible variants of the present definition of sliding, such as to slide with nodes at $(i-1,j)$ and $(i+1,j+1)$ or even with a single node at $(i-1,j)$ or at $(i+1,j)$. In this paper, though, we only focus on our original definition.}

\begin{figure}[!hbtp]
\centering{
\includegraphics[width=0.5\textwidth]{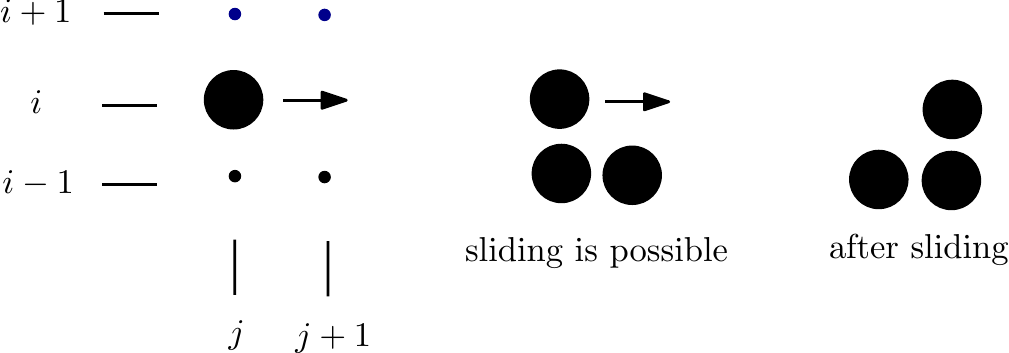}
}
\caption{Sliding to the right. Either the two blues (dots in row $i+1$) or the two blacks (dots in row $i-1$) and an empty cell at position $(i,j+1)$ are required for this movement. Then an example movement with the two blacks is given.} \label{fig:sliding-definition}
\end{figure}

Let $A$ and $B$ be two shapes. We say that \emph{$A$ transforms to $B$ via a movement $m$} (which can be either a rotation or a sliding), denoted $A \stackrel{m}\rightarrow B$, if there is a node $u$ in $A$ such that if $u$ applies $m$, then the shape resulting after the movement is $B$ (possibly after rotations and translations of the resulting shape, depending on the application). We say that $A$ \emph{transforms in one step to $B$} (or that \emph{$B$ is reachable in one step from $A$}), denoted $A\rightarrow B$, if $A \stackrel{m}\rightarrow B$ for some movement $m$. We say that $A$ \emph{transforms} to $B$ (or that $B$ is \emph{reachable} from $A$) and write $A\rsa B$, if there is a sequence of shapes $A=C_{0},C_{1},\ldots,C_{t}=B$, such that $C_{i}\rightarrow C_{i+1}$ for all $i$, $0\leq i <t$. We should mention that we do not always allow $m$ to be any of the two possible movements. In particular, in Sections \ref{sec:rotation} and \ref{sec:rotation-connectivity} we only allow $m$ to be a rotation, as we there restrict attention to systems in which only rotation is available. We shall clearly explain what movements are permitted in each part of the paper.

\begin{proposition} \label{pro:par-equiv}
The relation ``transforms to'' (i.e., `$\rsa$') is a partial equivalence relation.
\end{proposition}
\begin{proof}
The relation `$\rsa$' is a binary relation on shapes. To show that it is a partial equivalence relation, we have to show that it is symmetric and transitive.

For symmetricity, we have to show that for all shapes $A$ and $B$, if $A\rsa B$ then $B\rsa A$. It suffices to show that for all $A,B$, if $A\ra B$ then $B\ra A$, meaning that every one-step transformation (which can be either a single rotation or a single sliding) can be \emph{reversed}. For the rotation case, this follows by observing that a rotation of a node $u$ can be performed iff there are two consecutive empty positions in its trajectory. When $u$ rotates, it leaves its previous position empty, thus, leaving in this way two consecutive positions empty for the reverse rotation to become enabled. The argument for sliding is similar.

For transivity, we have to show that for all shapes $A$, $B$, and $C$, if $A\rsa B$ and $B\rsa C$ then $A\rsa C$. By definition, $A\rsa B$ if there is a sequence of shapes $A=C_{0},C_{1},\ldots,C_{t}=B$, such that $C_{i}\rightarrow C_{i+1}$ for all $i$, $0\leq i <t$ and $B\rsa C$ if there is a sequence of shapes $B=C_{t},C_{t+1},\ldots,C_{t+l}=C$, such that $C_{i}\rightarrow C_{i+1}$ for all $i$, $t\leq i <t+l$. So, for the sequence $A=C_{0},C_{1},\ldots,C_{t}=B,C_{t+1},\ldots,C_{t+l}=C$ it holds that $C_{i}\rightarrow C_{i+1}$ for all $i$, $0\leq i <t+l$, that is, $A\rsa C$.
\end{proof}

When the only available movement is rotation, there are shapes in which no rotation can be performed (we will see such examples in Section \ref{sec:rotation}). If we introduce a \emph{null} rotation, then every shape may transform to itself by applying the \emph{null} rotation. That is, reflexivity is also satisfied, and, together with symmetricity and transivity from Proposition \ref{pro:par-equiv}, ``transforms to'' (by rotations only) becomes an equivalence relation.

\begin{definition} \label{def:perimeter1}
Let $A$ be a connected shape. Color black each cell of the grid that is occupied by a node of $A$. A cell $(i,j)$ is part of a \emph{hole} of $A$ if every infinite length single path starting from $(i,j)$ (moving only horizontally and vertically) necessarily goes through a black cell. Color black also every cell that is part of a hole of $A$, to obtain a \emph{compact} black shape $A^\prime$ (i.e., one with no holes in it). Consider now polygons defined by unit-length line segments of the grid. Define the \emph{perimeter} of $A$ as the minimum-area such polygon that completely encloses $A^\prime$ in its interior. The fact that the polygon must have an \emph{interior} and an \emph{exterior} follows directly from the Jordan curve theorem \cite{Jo93}.
\end{definition}

\begin{definition} \label{def:perimeter2}
Now, color red any cell of the grid that has contributed at least one of its line-segments to the perimeter and is not black (i.e., is not occupied by a node of $A$). Call this the \emph{cell-perimeter} of shape $A$. See Figure \ref{fig:perimeter-definition} for an example.
\end{definition}

\begin{figure}[!hbtp]
\centering{
\includegraphics[width=0.35\textwidth]{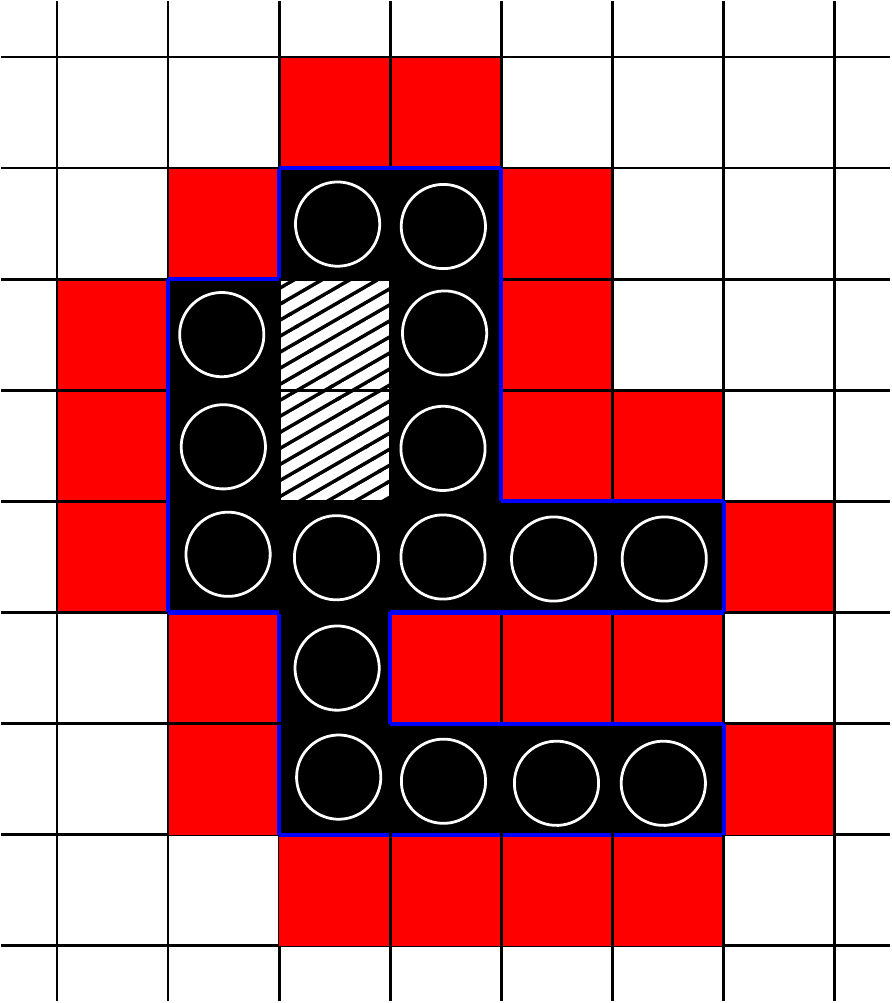}
}
\caption{The perimeter (polygon of unit-length line segments colored blue) and the cell-perimeter (cells colored red) of a shape $A$ (white spherical nodes; their corresponding cells have been colored black). The dashed black cells correspond to a hole of $A$.} \label{fig:perimeter-definition}
\end{figure}

\begin{definition} \label{def:perimeter3}
The \emph{external surface} of a connected shape $A$, is a shape $B$, not necessarily connected, consisting of all nodes $u\in A$ such that $u$ occupies a cell defining at least one of the line-segments of $A$'s perimeter.
\end{definition}

\begin{definition} \label{def:perimeter4}
The \emph{extended external surface} of a connected shape $A$, is defined by adding to $A$'s external surface all nodes of $A$ whose cell shares a corner with $A$'s perimeter (for example, the black node just below the hole, in Figure \ref{fig:perimeter-definition}).
\end{definition}

\begin{proposition} \label{pro:extended-external-connected}
The extended external surface of a connected shape $A$, is itself a connected shape.
\end{proposition}
\begin{proof}
The perimeter of $A$ is connected, actually, it is a cycle. This connectivity is preserved by the extended external surface, as whenever the perimeter moves straight, we have two horizontally or vertically neighboring nodes on the extended external surface and whenever it makes a turn, we either stay put or preserve connectivity via an intermediate diagonal node (from those nodes used to extend the external surface).
\end{proof}

Observe, though, that the extended external surface is not necessarily a cycle. For example, the extended external surface of a line-shape is equal to the shape itself (and, therefore, a line).

\subsection{Problem Definitions}
\label{subsec:problems}

We here provide formal definitions of all the transformation problems that are considered in this work.\\

\noindent\textbf{{\sc Rot-Transformability}.} Given an initial shape $A$ and a target shape $B$ (usually both connected), decide whether $A$ can be transformed to $B$ (usually, under translations and rotations of the shapes) by a sequence of rotation only movements.\\

\noindent\textbf{{\sc RotC-Transformability}.} The special case of {\sc Rot-Transformability} in which $A$ and $B$ are connected shapes and, additionally, connectivity must be preserved throughout the transformation.\\

\noindent\textbf{{\sc RS-Transformability}.} Given an initial shape $A$ and a target shape $B$ (usually both connected), decide whether $A$ can be transformed to $B$ (usually, under translations and rotations of the shapes) by a sequence of rotation and sliding movements.\\

\noindent\textbf{Minimum-Seed-Determination.} Given an initial shape $A$ and a target shape $B$ (usually only with rotation available and a proof that $A$ and $B$ are not transformable to each other without additional assumptions) determine a minimum-size seed and an initial positioning of that seed relative to $A$ that makes the transformation from $A$ to $B$ feasible. There are several meaningful variations of this problem. For example, the seed may or may not form part of the target shape or the seed may be used as an intermediated step to show feasibility with ``external'' help and then be able to show that, instead of externally providing it, it is possible to \emph{extract} it from the initial shape $A$ via a sequence of moves. We will clearly indicate which version is considered in each case.\\

In the above problems, the goal is to show feasibility of a set of transformation instances and, if possible, to provide an algorithm that decides feasibility. \footnote{An immediate next goal is to devise an algorithm able to compute an actual transformation or even compute or approximate the \emph{optimum} transformation (usually with respect to the number of moves). We leave these as interesting open problems.}

In the last part of the paper, we consider \emph{distributed transformation tasks}. There, the nodes are \emph{distributed processes} able to perform \emph{communicate-compute-move rounds} and the goal is to program them so that they (algorithmically) self-transform their initial arrangement to a target arrangement.\\

\noindent\textbf{Distributed-Transformability.} Given an initial shape $A$ and a target shape $B$ (usually by having access to both rotation and sliding), the nodes (which are now distributed processes), starting from $A$, must transform themselves to $B$ by a sequence of communication-computation-movement rounds. In the distributed transformations, we mostly consider the case in which $A$ can be \emph{any connected shape} and $B$ is a \emph{spanning line}, i.e., a linear arrangement of all the nodes.

\section{Rotation}
\label{sec:rotation}

In this section, the only permitted movement is 90\degree\ \emph{rotation} around a neighbor.

Consider a black and red checkered coloring of the 2D grid, similar to the coloring of a chessboard. Then any shape $S$ may be viewed as a colored shape consisting of $b(S)$ blacks and $r(S)$ reds. Call two shapes $A$ and $B$ \emph{color-consistent} if $b(A)=b(B)$ and $r(A)=r(B)$ and call them \emph{color-inconsistent} otherwise. Call a transformation from a shape $A$ to a shape $C$ \emph{color-preserving} if $A$ and $C$ are color consistent. Observe now, that if $A\ra B$, then $A$ and $B$ are color-consistent, because a rotation can never move a node to a position of different color than its starting position. This implies that if $A\rsa C$, then $A$ and $C$ are color-consistent, because any two consecutive shapes in the sequence are color-consistent. We conclude that:

\begin{observation}
The rotation movement is color-preserving. Formally, $A\rsa C$ (restricted to rotation only) implies that $A$ and $C$ are color-consistent. In particular, every node beginning from a black (red) position of the grid, will always be on black (red, respectively) positions throughout a transformation consisting only of rotations.
\end{observation}

Based on this property of the rotation movement, we may call each node \emph{black} or \emph{red} throughout a transformation, based only on its initial coloring. The above observation gives a partial way to determine that two shapes $A$ and $B$ cannot be transformed to each other by rotations.

\begin{proposition} \label{impossibility-color-inc}
If two shapes $A$ and $B$ are color-inconsistent, then it is impossible to transform one to the other by rotations only.
\end{proposition}

We now show that the inverse is not true, that is, it does not hold that any two color-consistent shapes can be transformed to each other by rotations. This is trivial for disconnected shapes, as any collection of isolated nodes cannot move at all, and either we consider only the cardinalities of the colors, in which case any two such shapes of equal cardinalities correspond to the same shape, or we also consider the precise positions of the nodes on the grid (e.g. by their relative distances), in which case no two such shapes can be transformed to each other. Thus, we show a counterexample for the case of connected shapes. We begin with a proposition relating the number of black and red nodes in a connected shape.

\begin{proposition} \label{pro:color-num-bounds}
A connected shape with $k$ blacks has at least $\lceil (k-1)/3\rceil$ and at most $3k+1$ reds.
\end{proposition}
\begin{proof}
For the upper bound, observe that a black can hold up to 4 distinct reds in its neighborhood, which implies that $k$ blacks can hold up to $4k$ reds in total, even if the blacks were not required to be connected to each other. To satisfy connectivity, every black must share a red with some other black (if a black does not satisfy this, then it cannot be connected to any other black). Any such sharing reduces the number of reds by at least 1. As at least $k-1$ such sharings are required for each black to participate in a sharing, it follows that we cannot avoid a reduction of at least $k-1$ in the number of reds, which leaves us with at most $4k-(k-1)=3k+1$ reds.

For the lower bound, if we invert the roles of blacks and reds, we have that $l$ reds can hold at most $3l+1$ blacks. So, if $k$ is the number of blacks, it holds that $k\leq 3l+1 \Leftrightarrow l\geq (k-1)/3$ and due to the fact that the number of reds must be an integer, we conclude that for $k$ blacks the number of reds must be at least $\lceil (k-1)/3\rceil$.
\end{proof}

\begin{proposition} \label{pro:line-with-leaves}
There is a generic connected shape, called \emph{line-with-leaves}, that has a color-consistent version for any connected shape. In other words, for $k$ blacks it covers the whole range of reds from $\lceil (k-1)/3\rceil$ to $3k+1$ reds.
\end{proposition}
\begin{proof}
Consider a bi-color line starting with a black node and ending to a black node, such that all $k$ blacks are exhausted, as shown in Figure \ref{fig:saturated-line}. To do this, $k-1$ reds are needed in order to alternate blacks and reds on the line. Next, ``saturate'' every black (i.e. maximize its degree) by adding as many red nodes as it can fit around it (recall that the maximum degree of every node is 4). The resulting saturated shape has $k$ blacks and $3k+1$ reds. This shape covers the $3k+1$ upper bound on the possible number of reds. By removing red leaf-nodes (i.e., of degree 1) one after the other, we can achieve the whole range of numbers of reds, from $k-1$ to $3k+1$ reds. It suffices to restrict attention to the range from $k$ to $3k+1$ reds. Take now any connected shape $A$ and color it in such a way that red is the majority color, that is $l\geq k$, where $l$ is the number of reds and $k$ is the number of blacks (there is always a way to do that). From the upper bound of Proposition \ref{pro:color-num-bounds}, $l$ can be at most $3k+1$, so we have $k\leq l\leq 3k+1$ for any connected shape $A$, which falls within the range that the line-with-leaves can represent. Therefore, we conclude that any connected shape $A$ has a color-consistent shape $B$ from the line-with-leaves family.
\end{proof}

\begin{figure}[!hbtp]
\centering{
\includegraphics[width=0.4\textwidth]{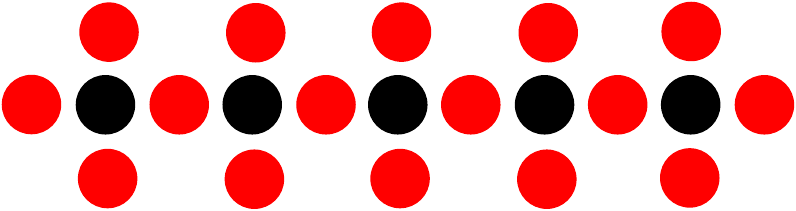}
}
\caption{A saturated line-with-leaves shape, in which there are $k=5$ blacks and $3k+1=16$ reds.} \label{fig:saturated-line}
\end{figure}

\begin{proposition} \label{pro:counterexample-transformability}
There is an infinite set of pairs $(A,B)$ of connected shapes, such that $A$ and $B$ are color-consistent but cannot be transformed to each other by rotations only.
\end{proposition}
\begin{proof}
For shape $A$, take a rhombus as shown in Figure \ref{fig:rhombus}, consisting of $k^2$ blacks and $(k+1)^2$ reds, for any $k\geq 2$. In this shape, every black node is ``saturated'', meaning that it has 4 neighbors, all of them necessarily red. This immediately excludes the blacks from being able to move, as all their neighboring positions are occupied by reds. But the same holds for the reds, as all potential target-positions for a rotation are occupied by reds. Thus, no rotation movement can be applied to any such shape $A$ and $A$ can only be transformed to itself (by \emph{null} rotations). By Proposition \ref{pro:line-with-leaves}, any such $A$ has a color-consistent shape $B$ from the family of line-with-leaves shapes, such that $B$ is not equal to $A$ (actually in $B$ several blacks may have degree 3 in contrast to $A$ where all blacks have degree 4). We conclude that $A$ and $B$ are distinct color-consistent shapes which cannot be transformed to each other, and there is an infinite number of such pairs, as the number $k^2$ of black nodes of $A$ can be made arbitrarily large.
\end{proof}

\begin{figure}[!hbtp]
\centering{
\includegraphics[width=0.25\textwidth]{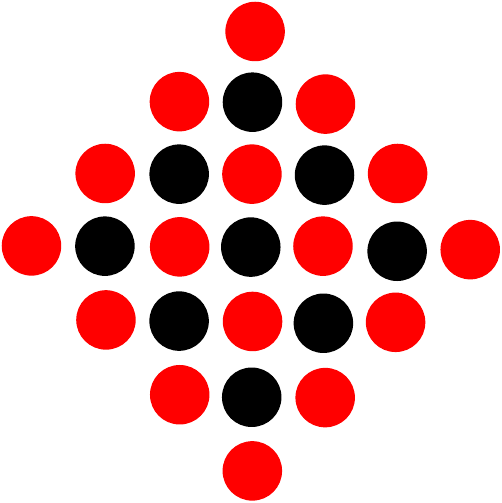}
}
\caption{A rhombus shape, consisting of $k^2=9$ blacks and $(k+1)^2=16$ reds.} \label{fig:rhombus}
\end{figure}

Propositions \ref{impossibility-color-inc} and \ref{pro:counterexample-transformability} give a partial characterization of pairs of shapes that cannot be transformed to each other. Observe that the impossibilities proved so far, hold for all possible transformations based on rotation only, i.e., they do not restrict the transformation in any way as would be, for example, to not allow the transformation to break the connectivity of the shape at any time.

A small shape of particular interest is a bi-color pair or \emph{2-line}. Such pairs can move easily in any direction, which makes them very useful components of transformations. One way to simplify some transformations would be to identify as many such pairs as possible in a shape and treat them in a different way than the rest of the nodes. A question in this respect is whether all the minority-color nodes of a connected shape can be completely to (distinct) nodes of the majority color. We show that this is not true.

\begin{proposition} \label{pro:matching-counterexample}
There is an infinite family of connected shapes, such that, if $A$ is a shape in the family of size $n$, then any matching of $A$ leaves at least $n/8$ nodes of each color unmatched.
\end{proposition}
\begin{proof}
See Figure \ref{fig:matching-counterexample}.
\end{proof}

\begin{figure}[!hbtp]
\centering{
\includegraphics[width=0.4\textwidth]{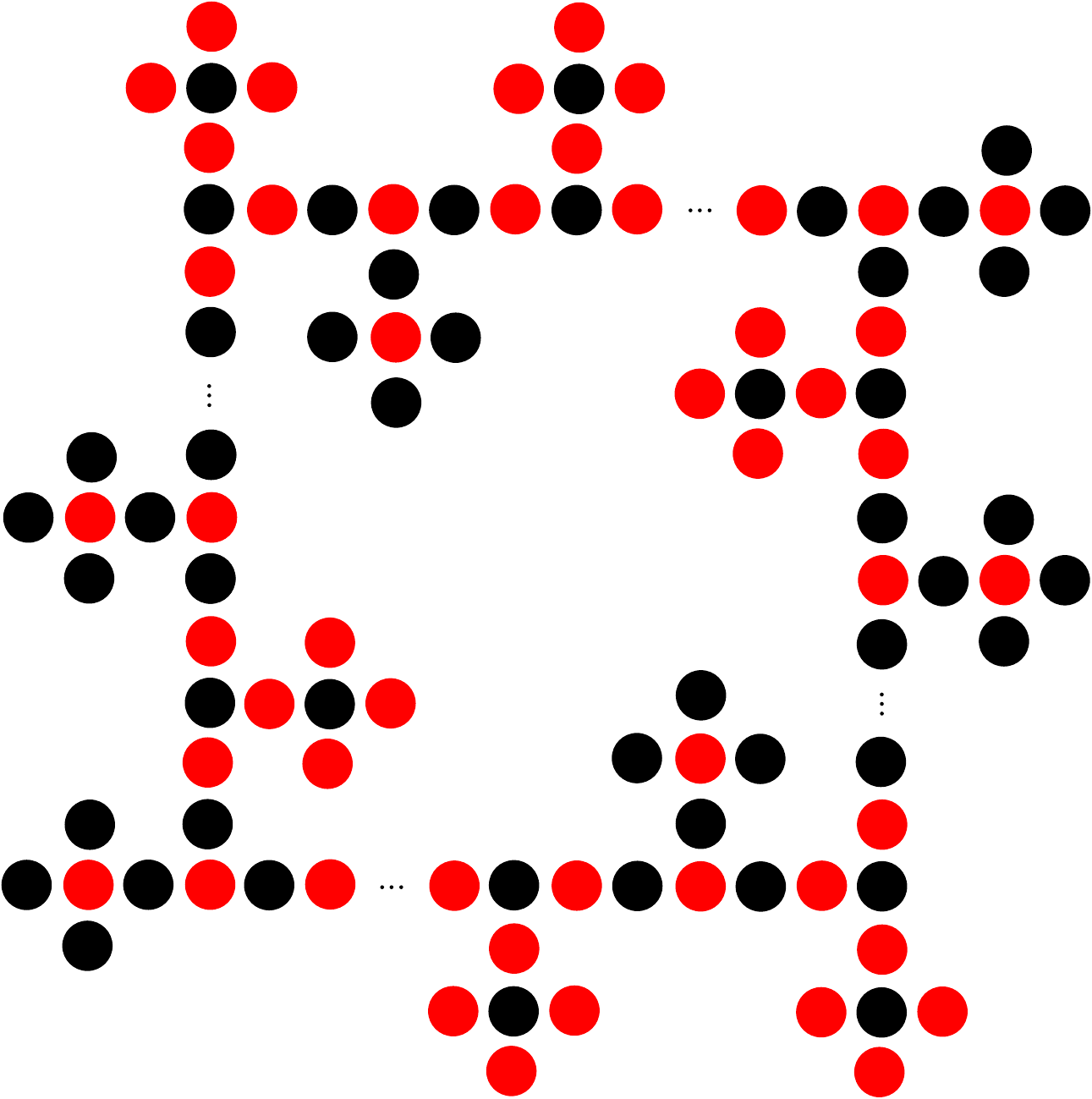}
}
\caption{The counterexample.} \label{fig:matching-counterexample}
\end{figure}

Recall that {\sc Rot-Transformability} is the language of all transformation problems between connected shapes that can be solved by rotation only and {\sc RotC-Transformability} is its subset obtained by the restriction that the transformation should not break the connectivity of the shape at any point during the transformation. We begin by showing that the inclusion between the two languages is strict, that is, there are strictly more feasible transformations if we allow connectivity to break. We prove that by showing that there is a feasible transformation in {\sc Rot-Transformability}$\setminus${\sc RotC-Transformability}.

\begin{theorem} \label{the:rotc-rot}
{\sc RotC-Transformability} $\subset$ {\sc Rot-Transformability}.
\end{theorem}
\begin{proof}
{\sc RotC-Transformability} $\subseteq$ {\sc Rot-Transformability} is immediate, as any transformation that does not break the shape's connectivity is also a valid transformation for {\sc Rot-Transformability}. So, it suffices to prove that there is a transformation problem in {\sc Rot-Transformability}$\setminus${\sc RotC-Transformability}. Consider a (connected) horizontal line of any even length $n$, and let $u_1,u_2,\ldots,u_n$ be its nodes. The transformation asks to fold the line onto itself, forming a double-line of length $n/2$ and width 2, i.e., a $n/2\times 2$ rectangle.

It is easy to observe that this problem is not in {\sc RotC-Transformability} for any $n>4$: the only nodes that can rotate without breaking connectivity are $u_1$ and $u_n$, but any of their two possible rotations only enables a rotation that will bring the nodes back to their original positions. This means that, if the transformation is not allowed to break connectivity, then such a shape is trapped in a loop in which only the endpoints can rotate between three possible positions, therefore it is impossible to fold a line of length greater than 4.

On the other hand, if connectivity can be broken, we can perform the transformation by the following simple procedure, consisting of $n/4$ phases: In the beginning of every phase $i\in \{1,2,\ldots,\lfloor n/4\rfloor\}$, pick the nodes $u_{2i-1}, u_{2i}$, which shall at that point be the two leftmost nodes of the original line. Rotate $u_{2i-1}$ once clockwise, to move above $u_{2i}$, then $u_{2i}$ three times clockwise to move to the right of $u_{2i-1}$ (the first of these three rotations breaks connectivity and the third restores it), and then rotate $u_{2i-1}$ twice clockwise to move to the right of $u_{2i}$, then $u_{2i}$ twice clockwise to move to the right of $u_{2i-1}$ and repeat this alternation until the pair that moves to the right meets the previous pair, which will be when $u_{2i-1}$ becomes the left neighbor of $u_{2i-2}$ on the upper line of the rectangle under formation, or, in case $i=1$, when $u_{2i-1}$ goes above $u_n$ (see Figure \ref{fig:line-folding}). If $n/4$ is not an integer, then perform a final phase, in which the leftmost node of the original line is rotated once clockwise to move above its right neighbor, and this completes folding.
\end{proof}

\begin{figure}[!hbtp]
\centering{
\includegraphics[width=0.4\textwidth]{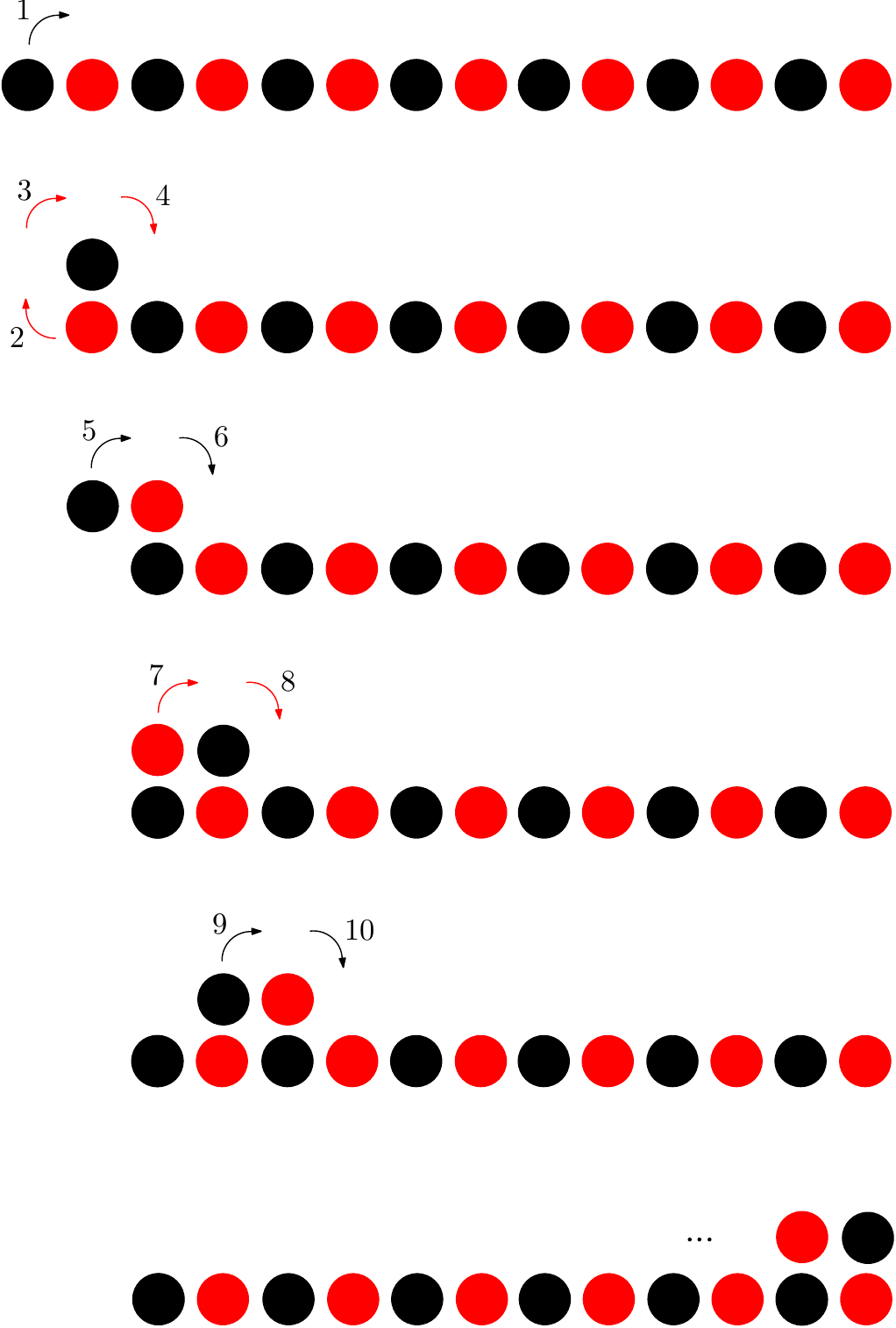}
}
\caption{Line folding.} \label{fig:line-folding}
\end{figure}

This means that allowing the connectivity to break enables more transformations, and this motivates us to start from this simpler case. But we already know from Proposition \ref{pro:counterexample-transformability}, that even in this case an infinite number of pairs of shapes cannot be transformed to each other. Aiming at a general transformation, we ask whether there is some minimal addition to a shape that would allow it to transform. The solution turns out to be as small as a \emph{2-line seed} lying initially somewhere ``outside'' the boundaries of the shape (e.g., just below the lowest row occupied by the shape).

Based on the above assumptions, we shall now prove that any pair of color-consistent connected shapes $A$ and $B$ can be transformed to each other. Recall from the discussion before Proposition \ref{pro:matching-counterexample}, that 2-line shapes can move freely in any direction. The idea is to use this 2-line in order to extract from the shape another 2-line, and use the two 2-lines together as a 4-line seed. The 4-line can also move freely in all directions. Then we shall use the 4-line as a transportation medium for those nodes that cannot move alone. In particular, we partition the nodes of the shape into those that can leave the shape as part of a 2-line and those that cannot. The latter nodes require the help of the 4-line to move them by carrying them, one at a time, in the form of a shape of order 5, which can only move diagonally (due to color-preservation of Proposition \ref{impossibility-color-inc}). We exploit these mobility mechanisms to transform $A$ into a uniquely defined shape from the line-with-leaves family of Proposition \ref{pro:line-with-leaves} (meaning that any two color-consistent shapes are matched to the same shape from the family). But if any connected shape $A$ with an extra 2-line can be transformed to its color-consistent line-with-leaves version with an extra 2-line, then this also holds inversely due to reversibility of rotations (discussed in the proof of Proposition \ref{pro:par-equiv}), and it follows that any $A$ can be transformed to any $B$ by transforming $A$ to its line-with-leaves version $L_A$ and then inverting the transformation from $B$ to $L_B=L_A$.

\begin{theorem} \label{the:rotation-generic}
If connectivity can break and there is a 2-line seed provided ``outside'' the initial shape, then any pair of color-consistent connected shapes $A$ and $B$ can be transformed to each other by rotations only.
\end{theorem}
\begin{proof}
Without loss of generality (due to symmetry and the 2-line's unrestricted mobility), it suffices to assume that the seed is provided somewhere below the lowest row $l$ occupied by the shape $A$. We show how $A$ can be transformed to $L_A$ with the help of the seed. We define $L_A$ as follows: Let $k$ be the cardinality of the minority color, let it be the black color. As there are at least $k$ reds, we can create a horizontal line of length $2k$, i.e., $u_1,u_2,\ldots,u_{2k}$, starting with a black, i.e., $u_1$ is black, and alternating blacks and reds. In this way, the blacks are exhausted. The remaining $\leq (3k+1)-k=2k+1$ reds are then added as leaves of the black nodes, starting from the position to the left of $u_1$ and continuing counterclockwise, i.e., below $u_1$, below $u_3$, ..., below $u_{2k-1}$, above $u_{2k-1}$, above $u_{2k-3}$, and so on. This gives the same shape from the line-with-leaves family, for all color-consistent shapes (observe that the leaf to the right of the line is always placed). $L_A$ shall be constructed on rows $l-5$ to $l-3$ (not necessarily inclusive), with $u_1$ on row $l-4$ and a column $j$ preferably between those that contain $A$.

First, extract a 2-line from $A$, from row $l$, so that the 2-line seed becomes a 4-line seed. To see that this is possible for every shape $A$ of order at least 2, distinguish the following two cases: (i) If the lowest row has a horizontal 2-line, then the 2-line can leave the shape without any help and approach the 2-seed. (ii) If not, then take any node $u$ of row $l$. As $A$ is connected and has at least two nodes, $u$ must have a neighbor $v$ above it. The only possibility that the 2-line $u$,$v$ is not free to leave $A$ is when $v$ has both a left and a right neighbor. Figure \ref{fig:extracting-2-line} shows how this can be resolved with the help of the 2-line seed (now the 2-line seed approaches and extracts the 2-line).

\begin{figure}[!hbtp]
\centering{
\includegraphics[width=0.6\textwidth]{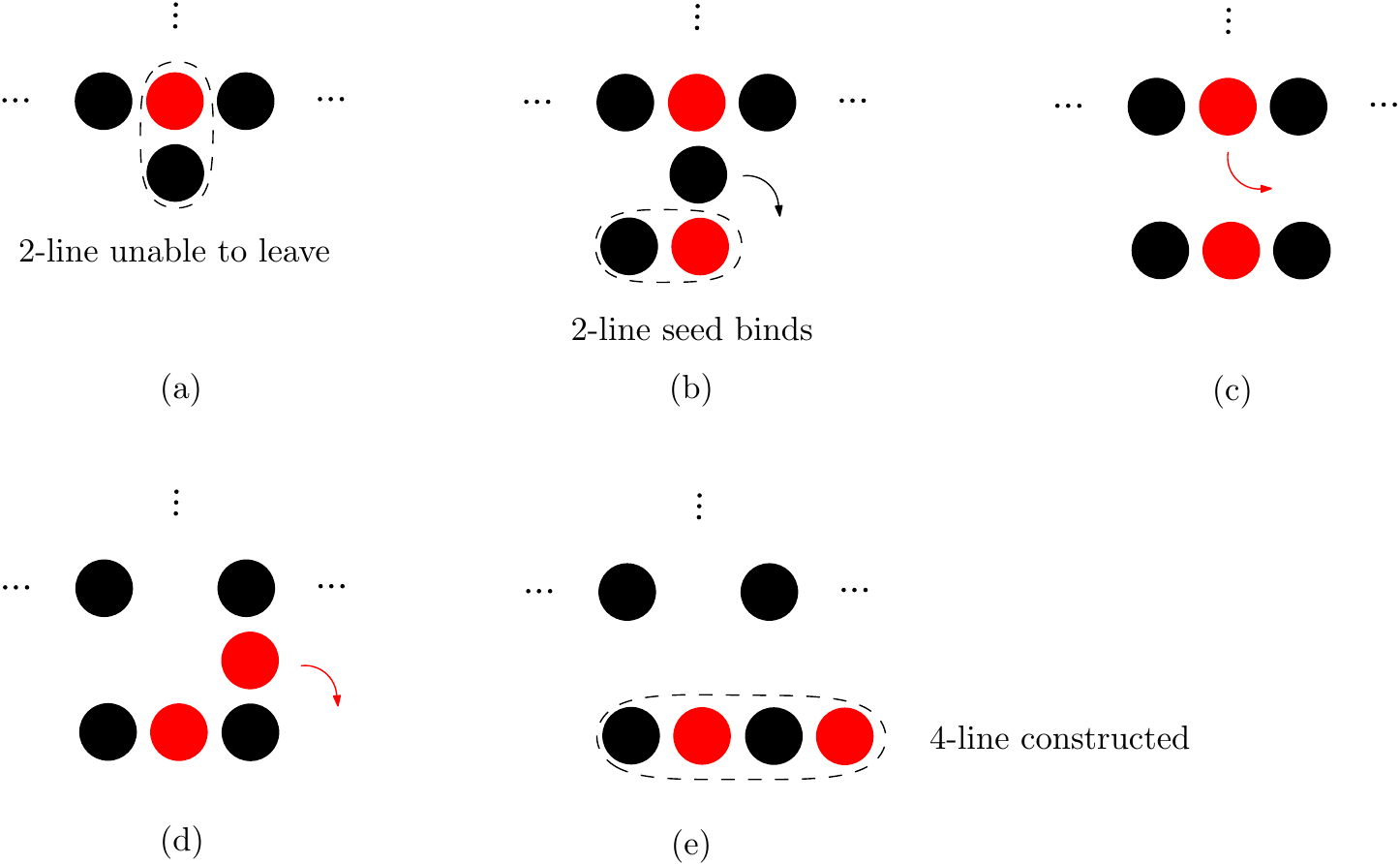}
}
\caption{Extracting a 2-line with the help of the 2-line seed.} \label{fig:extracting-2-line}
\end{figure}

To transform $A$ to $L_A$, given the 4-line seed, do the following:
\begin{itemize}
\item While the minority color (color chosen for $u_1$) is still present in $A$:
\begin{itemize}
\item If on the current lowest row occupied by $A$, there is a 2-line that can be extracted alone and move towards $L_A$, then perform the shortest such movement that attaches the 2-line to the right endpoint of $L_A$'s line $u_1, u_2, \ldots$.
\item If not, then use the 4-line to extract a single node from the lowest row of $A$. If that node fits to the right endpoint of $L_A$'s line, place it there, otherwise, transfer it to an unoccupied position below row $l-7$ to be used later.
\end{itemize}
\item Once the minority color has been exhausted from $A$, alternate the two colors until $u_{2k-3}$ has been placed ($u_{2k-1}$ and $u_{2k}$ will only be placed in the end as they are part of the 4-line). To do this, use the 4-line to transfer nodes from $A$ and from the ``repository'' maintained below $L_A$. When this occurs, if there are no more nodes left, run the termination phase, otherwise transfer the remaining nodes with the 4-line, one after the other, and attach them around the line of $L_A$, beginning from the position to the left of $u_1$ counterclockwise, as decribed above (skipping the position $u_{2k}$).
\item Termination phase: the line-with-leaves is ready, apart from positions $u_{2k-1}$, $u_{2k}$ which require a 2-line from the 4-line. If the position above $u_{2k-1}$ is empty, then extract a 2-line from the 4-line and transfer it to the positions $u_{2k-1}$, $u_{2k}$. This completes the transformation. If the position above $u_{2k-1}$ is occupied by a node $u_{2k+1}$, then place the whole 4-line vertically with its lowest endpoint on $u_{2k}$ (as in Figure \ref{fig:termination-4-line}). Then rotate the top endpoint counterclockwise, to move above $u_{2k+1}$, then rotate $u_{2k+1}$ clockwise around it to move to its left, then rotate the node above $u_{2k}$ counterclockwise to move to $u_{2k-1}$, and finally restore $u_{2k+1}$ to its original position. This completes the construction (the 2-line that always remains can be transferred in the end to a predefined position).
\end{itemize}
\end{proof}

\begin{figure}[!hbtp]
\centering{
\includegraphics[width=0.6\textwidth]{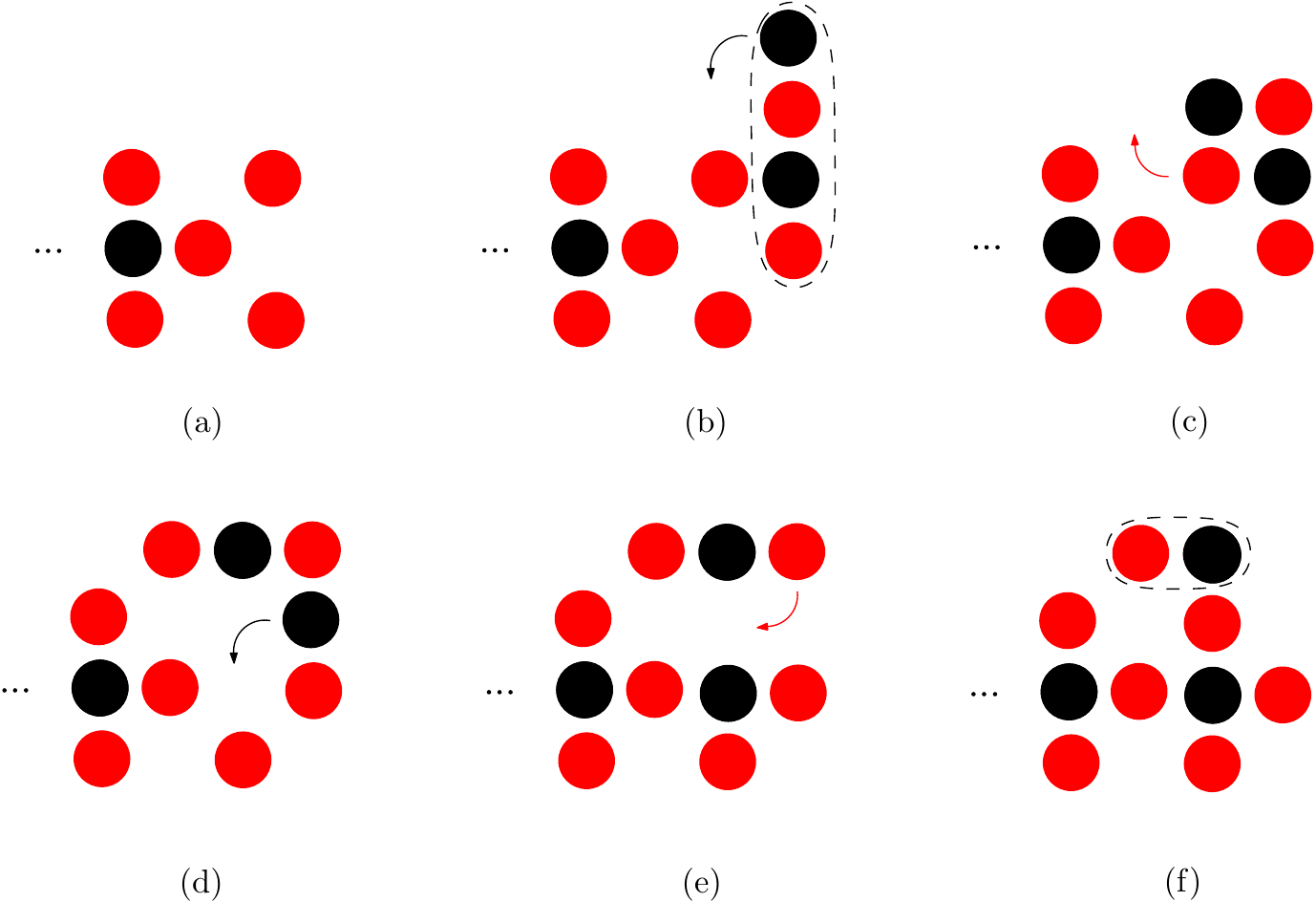}
}
\caption{The termination phase of the transformation.} \label{fig:termination-4-line}
\end{figure}

The natural next question is to what extent can the 2-line seed assumption be dropped. Clearly, by Proposition \ref{pro:counterexample-transformability}, this cannot be always possible. The following corollary gives a sufficient condition to drop the 2-line seed assumption, without looking deep into the structure of the shapes that satisfy it.

\begin{corollary}
Assume rotations only and that connectivity can break. Let $A$ and $B$ be two color-consistent connected shapes such that each one of them can self-extract a 2-line. Then $A$ and $B$ can be transformed to each other.
\end{corollary}

We remind that a rotation move in a grid can occur towards $4$ directions:  $NorthEast(1)$, $SouthEast(2)$, $SouthWest(3)$, $NorthWest(4)$. In order for the first move to occur a node has to be present North OR East but not both. The same requirements apply for moves $2$, $3$ and $4$ respectively.
If the connectivity of the shape can be broken and two nodes, A and B, are next to each other and A can perform a rotation using B, then B can perform a rotation using A if the connectivity of the shape can be broken.

\begin{lemma}\label{pro:move-2seed}
A 2-seed can be extracted from a shape iff a single rotation move is available on the shape.
\end{lemma}
\begin{proof}

If a move is available on a shape but not on the perimeter, that move can be transferred to the perimeter through transformations.
\\
\\
Let us consider a shape that has only two holes which are next to each other. We will call them cell $A$ and $E$ respectively. Without loss of generality let us consider that cell $A$ is west of $E$. We name $S$ the cell south of $A$ and $SE$ the cell south of $E$. Now we propose the following method. The node residing in cell $S$ rotates to the cell $E$ and then the node in cell $SE$ rotates to the cell $A$. After these two moves, cell $S$ is renamed to $A$ and $SE$ is renamed to $E$. The cell south of the new $A$ and $E$ are named $S$ and $SE$ respectively. This method can be repeated indefinitely until the two white cells reach the end of the grid.

We have shown how two white cells can``travel'' south. By reversing the method the two white cells can travel north. The two white cells can travel east and west with a simple transformation before the method. After naming the four cells above, the node in cell $S$ rotates to $E$. After this step we have two white cells, $A$ and $S$. Now rename $S$ into $E$ and $E$ into $S$. Now repeat the method and the two white positions will start travelling east. For the opposite direction, rotate the node in position $SE$ to cell $A$, and rename $A$ into $SE$ and $SE$ into $A$. Now repeat the method and the white cells can travel east. By using a combination of the above steps, the two white cells can move freely through the grid and reach any place.

Now consider a shape where there are more than $2$ holes but at least two are next to each other. We will show that the two white cells that are side by side can travel to the perimeter of the shape using the above method even if they reach other white cells. Without loss of generality suppose that the two white cells are the southernmost pair travelling south. If the travelling nodes ever meet a white cell south of them, we just need to show that we can turn this cell $(S)$ from a white one to a black one. Thus we perform the following act:
Check if there is a node west of $A$. If there is, move him south of $A$. Note that the cell west of $S$ is always a black node because we cannot have two white cells next to each other south of $A$. If not check if there is a node north of $A$. Note that there is always a node north of $A$, else a move would never be available which is prohibited. Now move the node north of $A$ to the west of $A$ then south of $A$. This move is available if there is a node northwest of $A$. If there is not, move the node north of $A$, east of $A$ then south of $A$. This move is available only if there is a node northeast of $A$. If there is not, move the node east to the northeastern cell of $A$ then east of $A$ then south of $A$.  If there is not one, we reach the following shape.See figure \ref{fig:polynomial}

The first node available northeast of $A$ or northwest of $A$ can be moved with rotations to the cell south of $A$. If a node is not available on either of those lines then either the connectivity of the shape is breached because we know that there are nodes north of $A$ which have to be connected with the rest of the shape, or the $A$ cell is not part of the shape. Both of those are not allowed so there is always a node northeast or northwest of $A$. Thus there is always a way to fill the cell south of $A$. In a similar fashion if the cell south of $E$ was white, we could always fill it.
\\
\\
A 2-seed can be extracted from a shape if a single rotation move is available on the perimeter of the shape.
\\
\\
Without the loss of generality suppose that nodes $A$ and $B$ are east-west to each other respectively, they are the southernmost nodes with a move available and none of them have any nodes in the two cells directly south of them. This means that the other can move as well. If node $A$ can perform a rotation to move south of $B$ then afterwards $B$ can perform a rotation to move west of $A$. Then $B$ can rotate south of $A$ and $A$ west of $B$. This four step method can be repeated forever until either one of them finds a node south. If one of them finds a node south, called $C$ i.e. $A$ find a node south of him then $B$ moves north of $A$ and $A$ moves east of $C$. Then $A$ and $C$ perform the four step method. If the two nodes keep repeating this eventually they will disconnect from the shape as a 2-seed.
\\
\\
If a move is not available a 2-seed cannot be extracted.
\\
\\
If a move is not available then no node can perform a rotation move. This means that no node can begin the process to extract himself as part of a 2-seed.
\end{proof}

\begin{figure}[!hbtp]
\centering{
\includegraphics[width=0.4\textwidth]{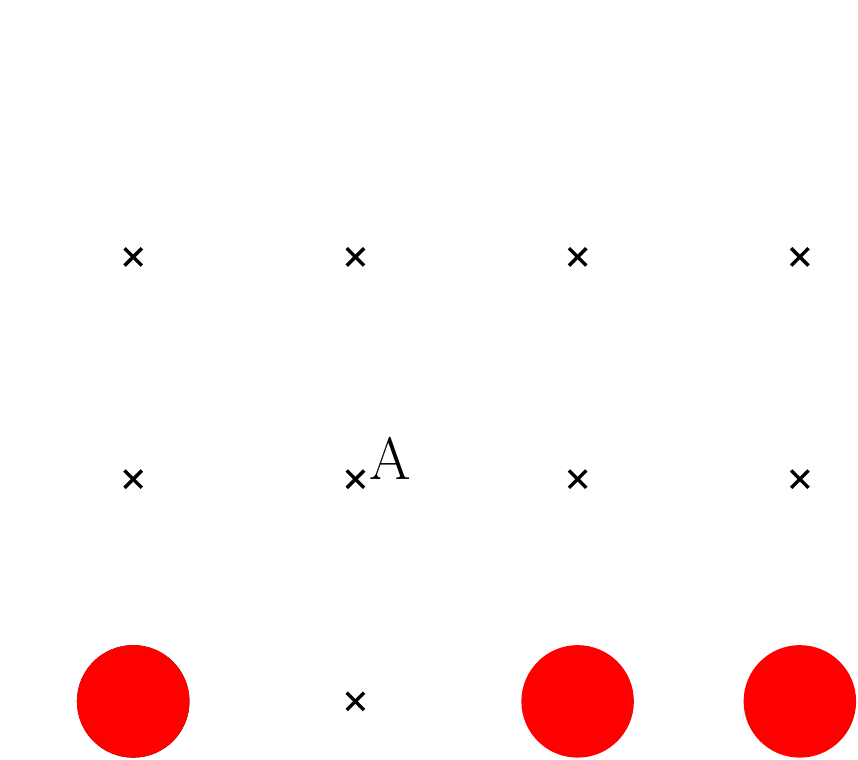}
}
\caption{Proof for 2seed extraction to the perimeter.}\label{fig:polynomial}
\end{figure}

\begin{theorem}
{\sc Rotation-Transformability} belongs to $\rem{P}$.
\end{theorem}
\begin{proof}
In Lemma \ref{pro:move-2seed}, we proved that we can extract a 2-seed from a shape iff a move is initially available. By Theorem \ref{the:rotation-generic}, if both shapes $A$ and $B$ have a 2-seed available then they can be transformed to each other. It follows that two shapes $A$ and $B$ can be transformed to each other iff both have a move available. Now we define a $n\times n$ grid where any shape with $n$ nodes can fit in. The time it takes for an algorithm to check if one of the shapes has a move available is $O(n)$. If for example the algorithm checks each individual node, that takes $O(1)$ time and, therefore, $O(n)$ time for $n$ nodes. So for two shapes it takes $O(n)$ time to check if a move is available in each of the shapes. Thus, the problem belongs to $\rem{P}$.

If the two shapes, $A$ and $B$, are the same, then they can trivially transform to each other without any moves. An algorithm can check this by simply mapping the grid of the first shape, which takes $O(n)$ time, and then check the second shape to see if the black cells match. If it ever finds a black cell that does not exist on the first shape, or it finds a white cell when it expected a black cell, then it decides that the two shapes are not the same. This process takes $O(n)$ time because it is equals to the time it takes to visit every $n$ node. Thus, it takes $O(n)$ time to check if $A=B$.
\end{proof}

\section{Rotation and Connectivity Preservation}
\label{sec:rotation-connectivity}

In this section, we restrict our attention to transformations that transform a connected shape $A$ to one of its color-consistent shapes $B$, without ever breaking the connectivity of the shape on the way. As already mentioned in the introduction, connectivity preservation is a very desirable property for programmable matter, as, among other positive implications, it guarantees that communication between all nodes is maintained, it minimizes transformation failures, requires less sophisticated actuation mechanisms, and increases the external forces required to break the system apart.

We begin by proving that {\sc RotC-Transformability} can be decided in deterministic polynomial space.

\begin{theorem}
{\sc RotC-Transformability} is in $\rem{PSPACE}$.
\end{theorem}
\begin{proof}
We first present a \emph{nondeterministic} Turing machine (NTM) $N$ that decides {\sc Transformability} in polynomial space. $N$ takes as input two shapes $A$ and $B$, both consisting of $n$ nodes and at most $4n$ edges. A reasonable representation is in the form of a binary $n\times n$ matrix (representing a large enough sub-area of the grid) where an entry is 1 iff the corresponding position is occupied by a node. Given the present configuration $C$, where $C=A$ initially, $N$ nondeterministically picks a valid rotation movement of a single node. This gives a new configuration $C^\prime$. Then $N$ replaces the previous configuration with $C^\prime$ in its memory, by setting $C\leftarrow C^\prime$. Moreover, $N$ maintains a counter $moves$ (counting the number of moves performed so far), with maximum value equal to the total number of possible shape configurations, which is at most $2^{n^2}$ in the binary matrix encoding of configurations. To set up such a counter, $N$ just have to reserve for it $n^2$ (binary) tape-cells, all initialized to 0. Every time $N$ makes a move, as above, after setting a value to $C^\prime$ it also increases $moves$ by 1, i.e., sets $moves\leftarrow moves+1$. Then $N$ takes another move and repeats. If it ever holds that  $C^\prime=B$ (may require $N$ to perform a polynomial-space pattern matching on the $n\times n$ matrix to find out), then $N$ accepts. If it ever holds that the counter is exhausted, that is, all its bits are set to 1, $N$ rejects. If $A$ can be transformed to $B$, then there must be a transformation beginning from $A$ and producing $B$, by a sequence of valid rotations, without ever repeating a shape. Thus, some branch of $N$'s computation will follow such a sequence and accept, while all non-accepting branches will reject after at most $2^{n^2}$ moves (when $moves$ reaches its maximum value). If $A$ cannot be transformed to $B$, then all branches will reject after at most $2^{n^2}$ moves. Thus, $N$ correctly decides {\sc Transformability}. Every branch of $N$, at any time, stores at most to shapes (the previous and the current), which requires $O(n^2)$ space in the matrix representation, and a $2^{n^2}$-counter which requires $O(n^2)$ bits. It follows that every branch uses space polynomial in the size of the input. So, far we have proved that {\sc Transformability} is decidable in nondeterministic polynomial (actually, linear) space. By applying Savitch's theorem \cite{Sa70} \footnote{Informally, Savitch's theorem establishes that any NTM that uses $f(n)$ space can be converted to a deterministic TM that uses only $f^2(n)$ space. Formally, it establishes that for any function $f\colon \bbbn\ra \bbbn$, where $f(n)\geq \log n$, $\rem{NSPACE}(f(n))\subseteq\rem{SPACE}(f^2(n))$.}, we conclude that {\sc Transformability} is also decidable in deterministic polynomial space (actually, quadratic), i.e., it is in $\rem{PSPACE}$.
\end{proof}

Recall that in the line folding problem, the initial shape is a (connected) horizontal line of any even length $n$, with nodes $u_1,u_2,\ldots,u_n$, and the transformation asks to fold the line onto itself, forming a double-line of length $n/2$ and width 2. As part of the proof of Theorem \ref{the:rotc-rot}, it was shown that if $n>4$, then it is impossible to solve the problem by rotation only (if $n=4$, it is trivially solved, just by rotating each endpoint above its unique neighbor). In the next proposition, we employ again the idea of a seed to show that with a little external help the transformation becomes feasible.

\begin{proposition}
If there is a 3-line seed $v_1,v_2,v_3$, horizontally aligned over nodes $u_3,u_4,u_5$ of the line, then the line can be folded.
\end{proposition}
\begin{proof}
We distinguish two cases, depending on whether we want the seed to be part of the final folded line or not. If yes, then we can either use a 4-line seed directly, over nodes $u_3,u_4,u_5,u_6$, or a 3-line seed but require $n$ to be odd (so that $n+3$ is even). If not, then $n$ must be even. We show the transformation for the first case, with $n$ odd and a 3-line seed (the other cases can be then treated with minor modifications).

We first show a simple reduction from an odd line with a 3-line seed starting over its third node to an even line with a 4-line seed starting over its third node. By rotating $u_1$ clockwise over $u_2$, we obtain the 4-line seed $u_1,v_1,v_2,v_3$. It only remains to move the whole seed two positions to the right (by rotating each of its 2-lines clockwise around themselves). In this manner, we obtain an even-length line $u_2,\ldots,u_n$ and a 4-line seed starting over its third node, without breaking connectivity. Therefore, in what follows we may assume that the initial shape is an even-length line $u_1,u_2,\ldots,u_n$ with a 4-line seed $v_1,v_2,v_3,v_4$ horizontally aligned over nodes $u_3,u_4,u_5,u_6$.

See Figure \ref{fig:line-folding-cp}.
\end{proof}

\begin{figure}[!hbtp]
\centering{
\includegraphics[width=0.5\textwidth]{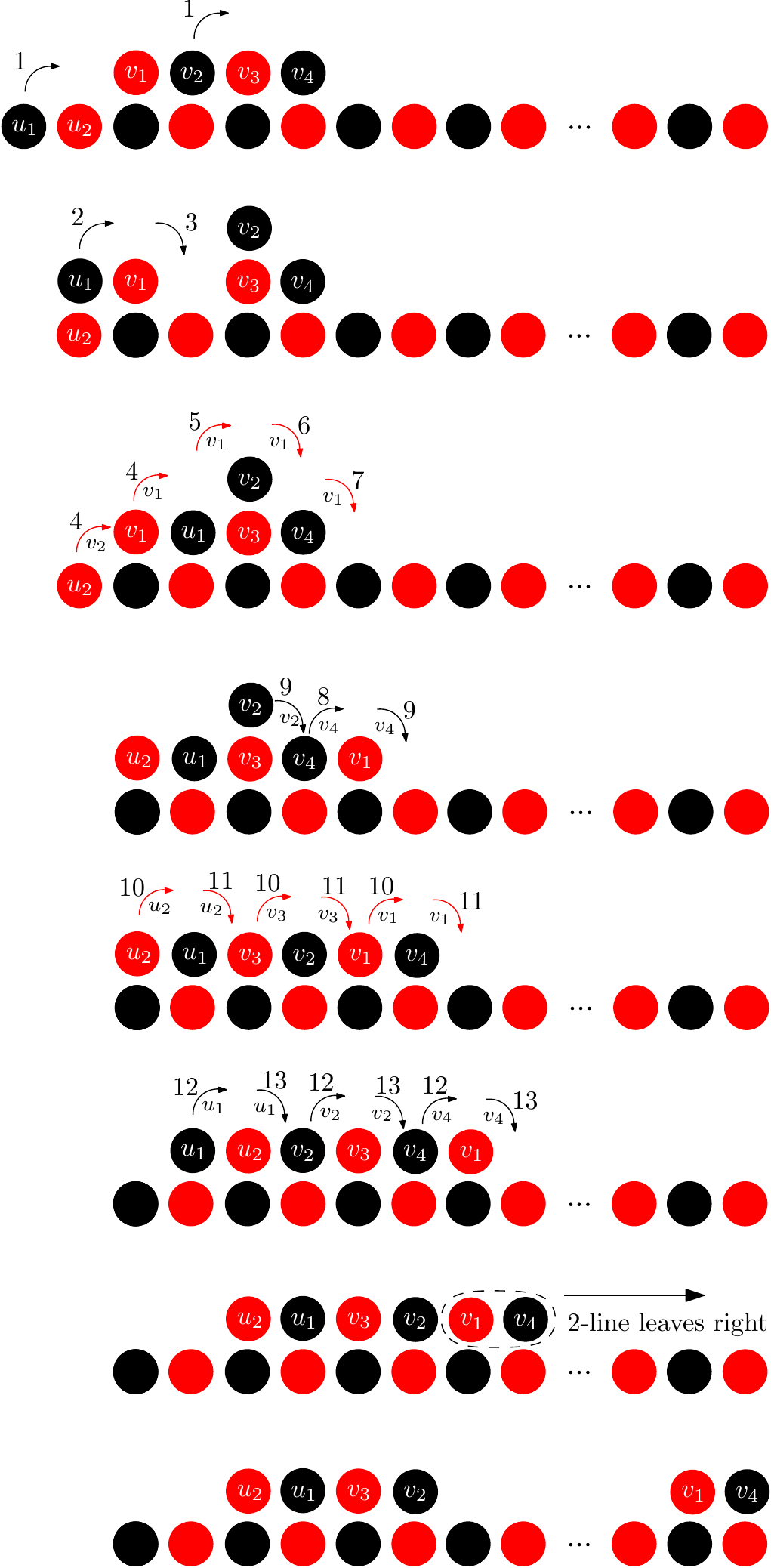}
}
\caption{The main subroutine of line folding with connectivity preservation.} \label{fig:line-folding-cp}
\end{figure}

We believe that in order to transform one shape to another we first need to find a seed that can both move on the perimeter of a shape and being able to reach every possible cell of the perimeter. We call this for simplicity: traverse the perimeter. After this is guaranteed we want the seed to be able to extract nodes and move them gradually to specific cells of the perimeter in order to create the desired shape. Thus the seed could actually simulate the rotation-sliding movement. We begin with the smallest seed possible and try to tackle the problem of moving on the perimeter of a line. Note that we do not allow the nodes of the shape to move in order to simplify and strengthen the model.

\begin{proposition}
A 2-seed cannot traverse the perimeter a line without breaking the connectivity.
\end{proposition}

\begin{proof}
Observe figure \ref{fig:2-seed-black}, shape number $1$. The 2-seed has reached the end of the line and now it tries to move east of the line and then south of it. Note that the black node has $2$ possible moves. It can either perform a single move and stop above the red node, or perform two subsequent moves and stop east of the red node. No matter the choice, the red node then is not able to move because any possible move would break the connectivity of the shape. See figure \ref{fig:2-seed-black}, shapes number $2,3$. Thus the black node has to stay in place and only the red node can move now. Observe now that the red node is trapped in a loop of $2$ possible moves (excluding the act of moving above the black node which would not allow us to try and move under the line): become the new endpoint or move under the end of the line. The first case leads necessarily to the second case because it is the only legal move available (excluding the move of looping back). See figure \ref{fig:2-seed-red}. But when we each the second case, once more we are limited into looping back to the initial positions. Thus a 2-seed cannot traverse the perimeter of a line without breaking the connectivity
\end{proof}

\begin{figure}[!hbtp]
\centering{
\includegraphics[width=1.0\textwidth]{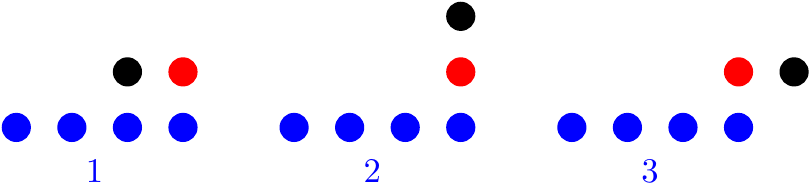}
}
\caption{Black node movement for 2-seed.}\label{fig:2-seed-black}
\end{figure}

\begin{figure}[!hbtp]
\centering{
\includegraphics[width=1.0\textwidth]{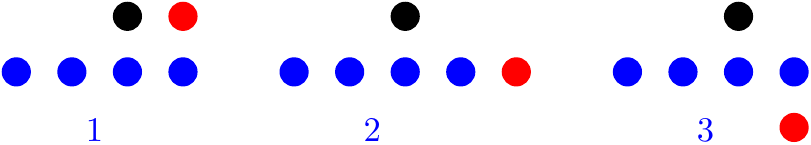}
}
\caption{Red node movement for 2-seed).}\label{fig:2-seed-red}
\end{figure}

\begin{proposition}
A 4-seed cannot traverse the perimeter of a line without breaking the connectivity.
\end{proposition}

\begin{proof}

Consider the last time, \emph{tlast}, that the black and red nodes in rows$\leq$i increases from $2$ to $3$. This means that either a black or a red moved at \emph{tlast} from $i+1$ to $i$. From now on, none of those $3$ nodes can go back to rows$>$i and there is one node remaining in rows$>$i. Actually that node u must necessarily be in row $i+1$, otherwise the connectivity would have broken. So no node from rows$\leq$i can return to rows$>$i anymore and there is a single node \emph{u} remaining in row $r+1$. We begin by finding the possible shapes that meet the above requirements.

The rotation of the node at \emph{tlast} was necessarily clockwise, as the closest counterclockwise move to the line is from $(i+1,j+1)$ to $(i,j)$, but it requires $(i,j)$ to be empty before rotating, but then $2$ nodes in rows$\leq$i and only one additional (u) in rows$>$i cannot support connectivity. We will now distinguish the \emph{tlast} into cases.

If u is a black node: If u is at position $(i+1,j-2)$ then it is stuck forever (blue node cannot move and the other black and red cannot go up any more to carry u. It also cannot be at $(i+1,j+2)$ as this does not permit a clockwise move of a red from $(i+1,j+1)$, so it has to be at $(i+1,j)$. See figure \ref{fig:4seed-tlast} shape number $2$. Node u is connected to A only via the red below it, which therefore cannot move unless u moves first(because no node can return to row$>$i any more to support u via another path. But the only way for u to move is for the black southeast node to move first, which in turn cannot move unless the rightmost red moves up which is impossible as no node may return to row$>$i (that red node can move down but then the only available movement is to return to its previous position.

If u is a red node: It cannot be at $(i+1,j-3)$ as before and it cannot be at $(i+1,j+1)$ as the rotation at \emph{tlast} was then necessarily from $(i+1,j)$ which is blocked by u. Observe that the clockwise rotation could not have been from $(i+1,j+2)$. The only way to support connectivity in this case with $2$ nodes in rows$\leq$i and $2$ in rows$>$i, is by having the following shape but then a clockwise rotation of the upper black is impossible. Therefore, if u is a red node it has to be at $(i+1,j-1)$. See figure \ref{fig:4seed-tlast} shape number $1$. Either nodes in rows$\leq$i cannot move at all, or if the bottom black is far away, the rightmost black is trapped in a loop going down and then up to its original positions as before.

Therefore the 4-seed cannot traverse the perimeter of a line without breaking the connectivity.

\qed
\end{proof}

\begin{figure}[!hbtp]
\centering{
\includegraphics[width=0.5\textwidth]{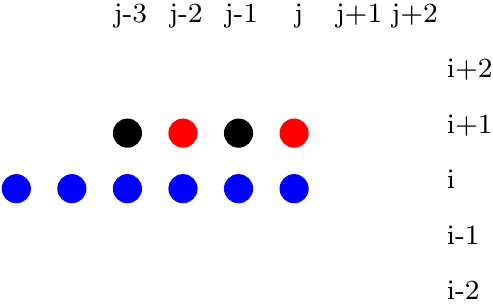}
}
\caption{Starting position of 4-seed on a line.}\label{fig:4seed-start}
\end{figure}

\begin{figure}[!hbtp]
\centering{
\includegraphics[width=1.0\textwidth]{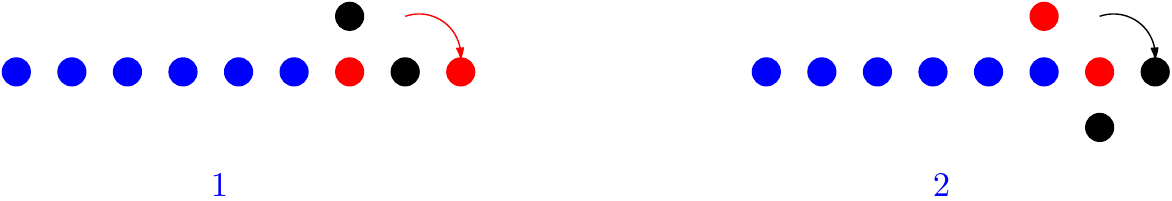}
}
\caption{Black and red cases for node u).}\label{fig:4seed-tlast}
\end{figure}

\begin{proposition}
A 6-seed can traverse the perimeter of a \emph{discrete-convex} shape without breaking the connectivity.
\end{proposition}

\begin{proof}
Consider a folded 6-seed occupying cells $(i,j)$, $(i,j+1)$, $(i,j+2)$ and $(i-1,j)$, $(i-1,j+1)$, $(i-1,j+2)$. Since the shape is discrete-convex, iff there is any node present in cells $(i,j+3)$ or $(i-1,j+3)$, there can be no node present in cells $(i,j-1)$ or $(i-1,j-1)$. For the same reason if there is a node in cells $(i-2,j)$ or $(i-2,j+1)$ or $(i-2,j+2)$ there can be no node in cells $(i+1,j)$ or $(i+1,j+1)$ or $(i+1,j+2)$.In order to place this seed at those cells, one of the neighbouring cells has to be occupied by a node. Without loss of generality, suppose that the 6-seed tries to move east. There are $4$ distinct cases for this move. Note that in the following $4$ cases we assume the absolute minimum amount of neighbouring nodes. If at any case there were more present at the shape, the rotations would be the exact same without any modification or problem.

A node occupies cell $(i,j+3)$. In order for this shape to be discrete-convex, a node has to be present in cell $(i-1,j+3)$. In this case the 6-seed has to move north and performs the rotations described in figure \ref{fig:orientation}, \ref{fig:orientation1}, \ref{fig:orientation2}, \ref{fig:orientation3} if a node is present in cell $(i+1,j+3)$ and the orientations described in figures \ref{fig:orientation-climb} if a node is not present in cell $(i+1,j+3)$ in order to keep moving.

A node occupies cell $(i-1,j+3)$ and no node occupies cell $(i,j+3)$. In this case the 6-seed performs the rotations described in figure \ref{fig:step} in order to \emph{climb} the step. Note that since the 6-seed begins and ends the move while preserving its shape, it is guaranteed that any number of steps can be climbed this way.

A node occupies cell $(i-2,j+3)$ and no node occupies cell $(i-1,j+3)$ or $(i,j+3)$. In this case the 6-seed performs the rotations described in figure \ref{fig:slide} rotations in order to \emph{slide} east.

No nodes occupy cells $(i-2,j+3)$, $(i-1,j+3)$ and $(i,j+3)$. In this case the 6-seed performs the rotations described in figure \ref{fig:dstep} in order to reach a shape that matches the conditions of the step case. Therefore the 6-seed can now perform a \emph{climb} move in order to continue.

We can replicate the results for south, west, north directions by simply rotating the whole shape by $90$, $180$, $270$ degrees respectively.
\end{proof}

\begin{figure}[!hbtp]
\centering{
\includegraphics[width=1.0\textwidth]{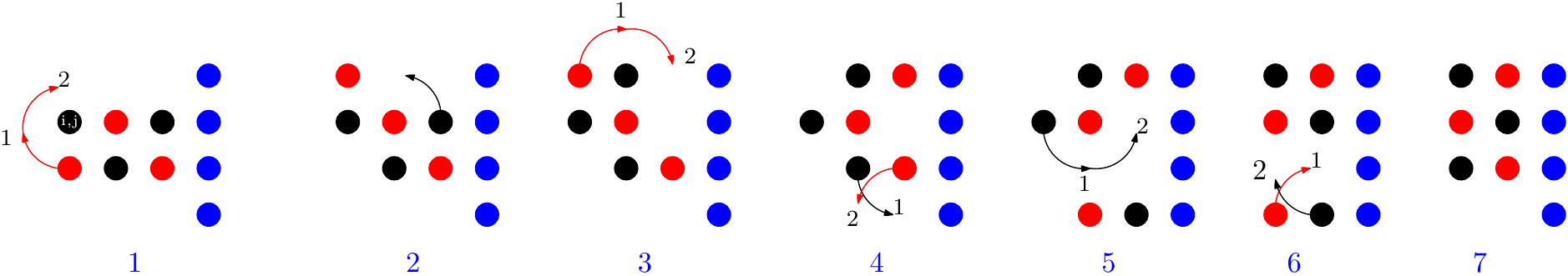}
}
\caption{Method for going north.}\label{fig:orientation}
\end{figure}

\begin{figure}[!hbtp]
\centering{
\includegraphics[width=1.0\textwidth]{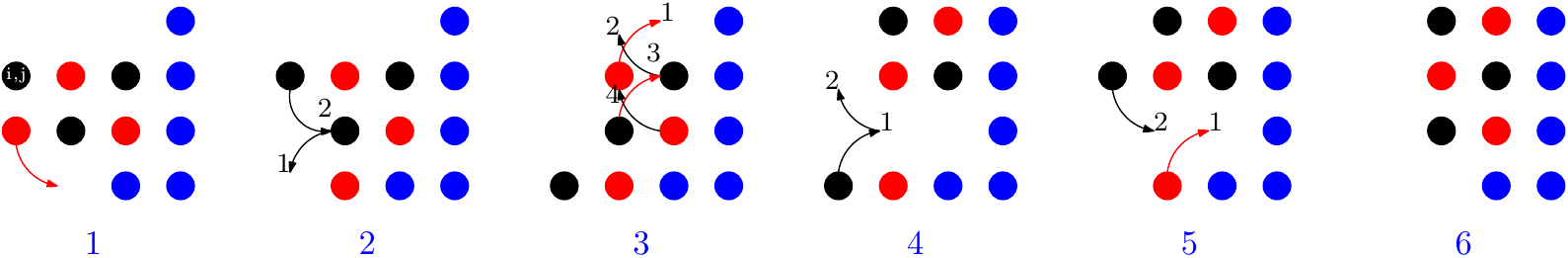}
}
\caption{Method for going north with node at (i-2,j+2).}\label{fig:orientation1}
\end{figure}

\begin{figure}[!hbtp]
\centering{
\includegraphics[width=1.0\textwidth]{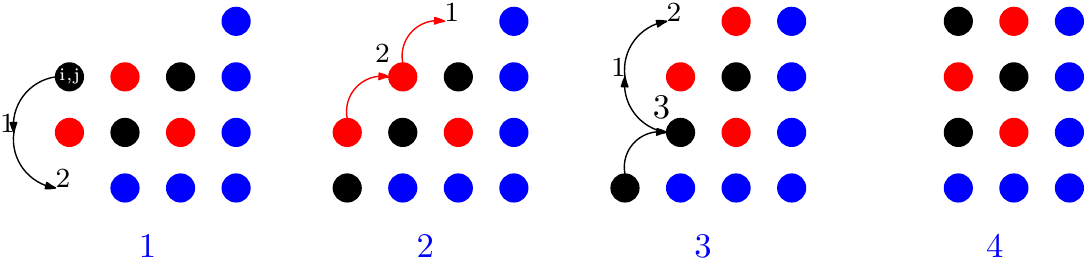}
}
\caption{Method for going north with node at (i-2,j+1).}\label{fig:orientation2}
\end{figure}

\begin{figure}[!hbtp]
\centering{
\includegraphics[width=1.0\textwidth]{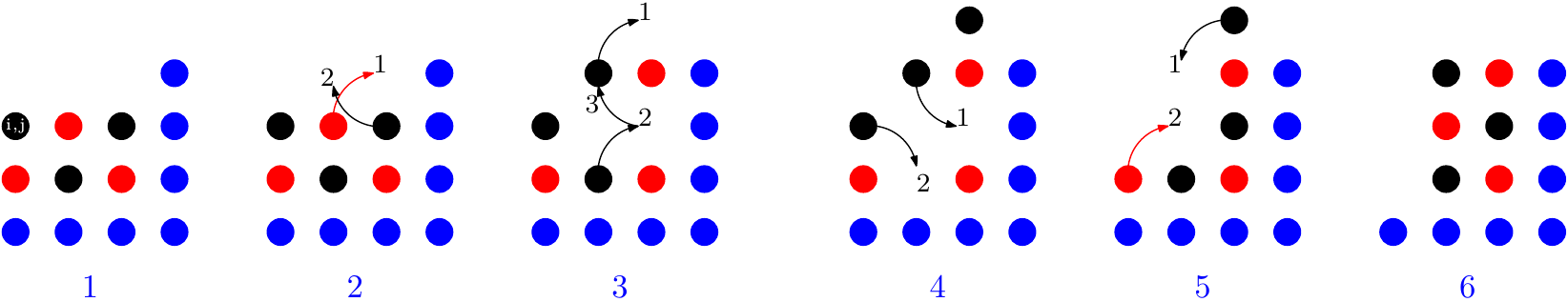}
}
\caption{Method for going north with node at (i-2,j).}\label{fig:orientation3}
\end{figure}

\begin{figure}[!hbtp]
\centering{
\includegraphics[width=1.0\textwidth]{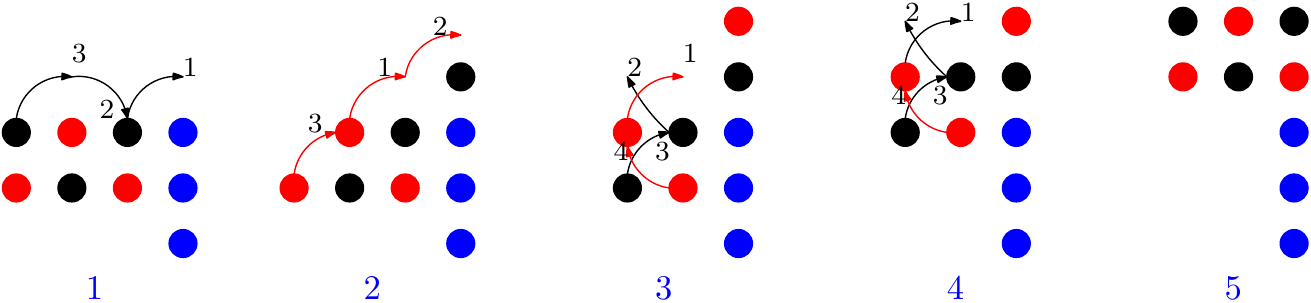}
}
\caption{Method for climbing north.}\label{fig:orientation-climb}
\end{figure}

\begin{figure}[!hbtp]
\centering{
\includegraphics[width=1.0\textwidth]{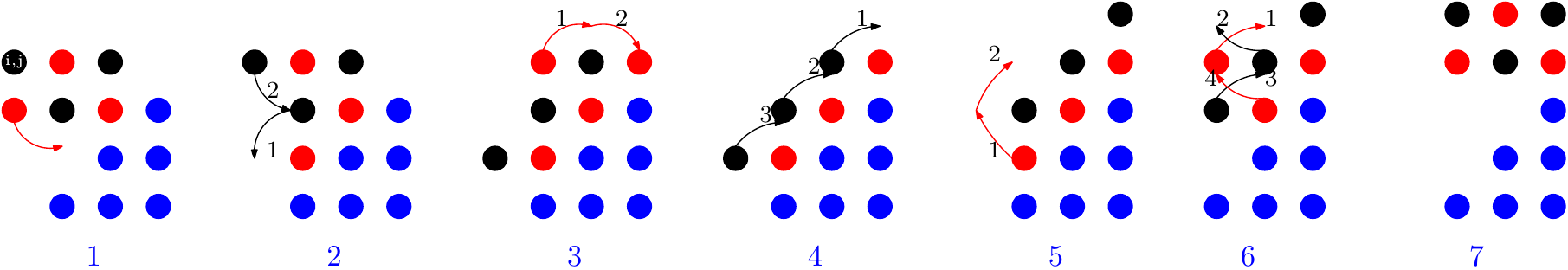}
}
\caption{Method for climbing a step.}\label{fig:step}
\end{figure}

\begin{figure}[!hbtp]
\centering{
\includegraphics[width=0.5\textwidth]{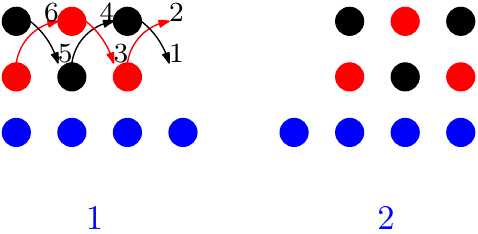}
}
\caption{Method for sliding.}\label{fig:slide}
\end{figure}

\begin{figure}[!hbtp]
\centering{
\includegraphics[width=1.0\textwidth]{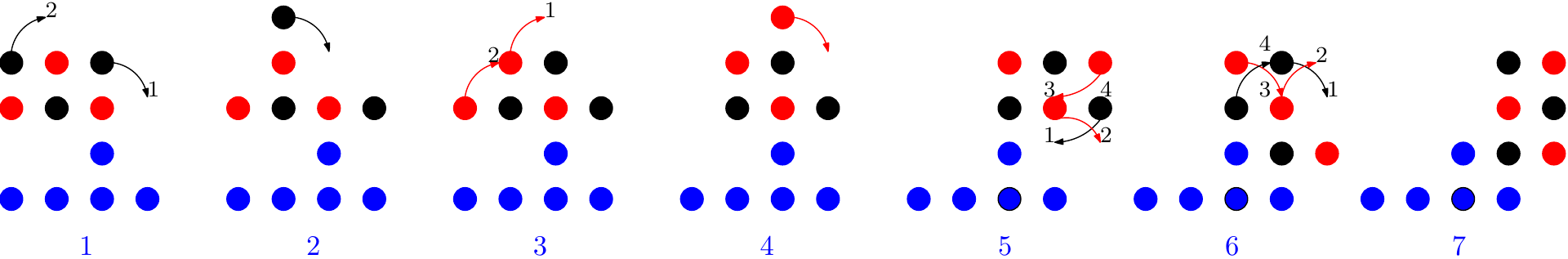}
}
\caption{Method for transforming.}\label{fig:dstep}
\end{figure}

\begin{proposition}
An 8-seed can traverse the perimeter of a \emph{discrete-convex} shape without breaking the connectivity.
\end{proposition}

\begin{proof}
Consider a folded 8-seed occupying cells $(i,j)$ $(i,j+1)$ $(i,j+2)$ $(i,j+3)$ and $(i-1,j)$, $(i-1,j+1)$, $(i-1,j+2)$, $(i-1,j+3)$. Since the shape is discrete-convex, iff there is any node present in cells $(i,j+4)$ or $(i-1,j+4)$, there can be no node present in cells $(i,j-1)$ or $(i-1,j-1)$. For the same reason if there is a node in cells $(i-2,j)$ or $(i-2,j+1)$ or $(i-2,j+2)$ or $(i-2,j+3)$ there can be no node in cells $(i+1,j)$ or $(i+1,j+1)$ or $(i+1,j+2)$ or $(i+1,j+3)$.In order to place this seed at those cells, one of the neighbouring cells has to be occupied by a node. Without loss of generality, suppose that the 8seed tries to move east. There are $4$ distinct cases for this move. Note that in the following $4$ cases, if not mentioned, we assume the absolute minimum amount of neighbouring nodes. If at any case there were more present at the shape, the rotations would be the exact same without any modification or problem.

A node occupies cell $(i,j+4)$. In order for this shape to be discrete-convex, a node has to present in cell $(i-1,j+4)$. In this case the 8seed has to move north and performs the rotations described in figure \ref{fig:8orientation} if a node is present in cell $(i+1,j+4)$; and the orientations described in figure \ref{fig:8orientation-climb} if a node is not present in cell $(i+1,j+4)$ in order to keep moving.

A node occupies cell $(i-1,j+4)$ and no node occupies cell $(i,j+4)$. In this case the 8-seed performs the rotations described in figure \ref{fig:8step} and in figure \ref{fig:8step1} order to \emph{climb} the step. Note that since the 8-seed begins and ends the move while preserving its shape, it is guaranteed that any number of steps can be climbed this way.

A node occupies cell $(i-2,j+4)$ and no node occupies cell $(i-1,j+4)$ or $(i,j+4)$. In this case the 8-seed performs the rotations described in figure \ref{fig:8slide} rotations in order to \emph{slide} east.

No nodes occupy cells $(i-2,j+4)$, $(i-1,j+4)$ and $(i,j+4)$. In this case the 8-seed performs the rotations described in figure \ref{fig:8dstep} in order to reach a shape that matches the conditions of the first case. Therefore the 8-seed can now perform a \emph{climb} move in order to continue.

We can replicate the results for south, west, north directions by simply rotating the whole shape by $90$, $180$, $270$ degrees respectively.
\end{proof}

\begin{figure}[!hbtp]
\centering{
\includegraphics[width=1.0\textwidth]{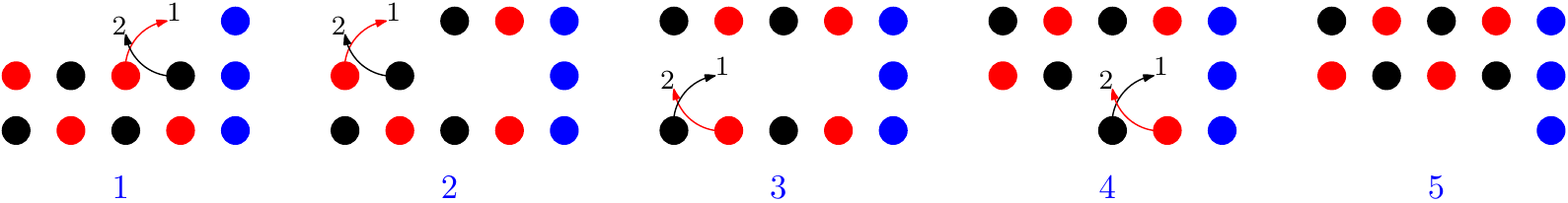}
}
\caption{Method for going north.}\label{fig:8orientation}
\end{figure}

\begin{figure}[!hbtp]
\centering{
\includegraphics[width=1.0\textwidth]{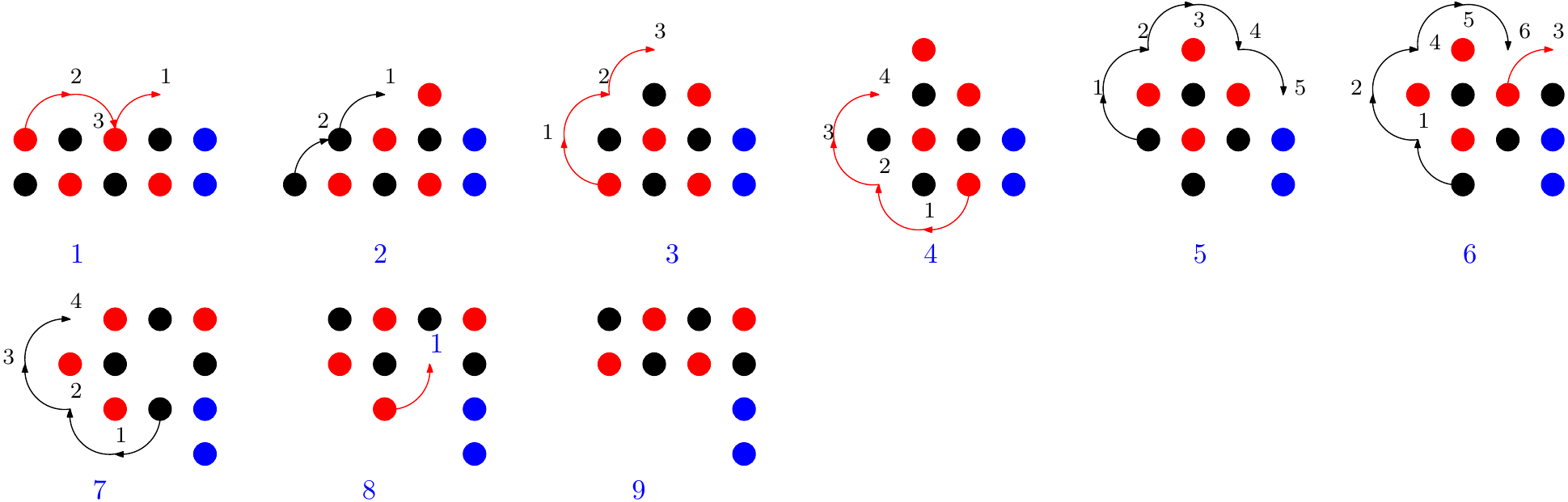}
}
\caption{Method for climbing north.}\label{fig:8orientation-climb}
\end{figure}

\begin{figure}[!hbtp]
\centering{
\includegraphics[width=1.0\textwidth]{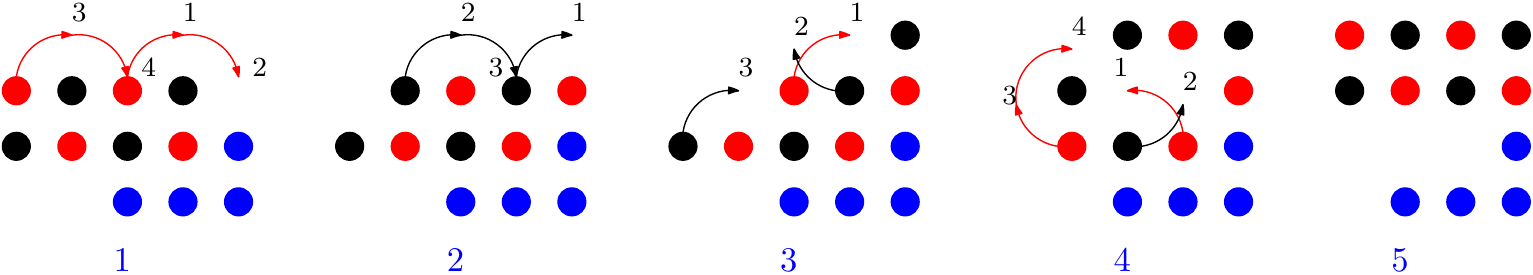}
}
\caption{Method for climbing a step with a node at cell (i-2,j+2).}\label{fig:8step}
\end{figure}

\begin{figure}[!hbtp]
\centering{
\includegraphics[width=1.0\textwidth]{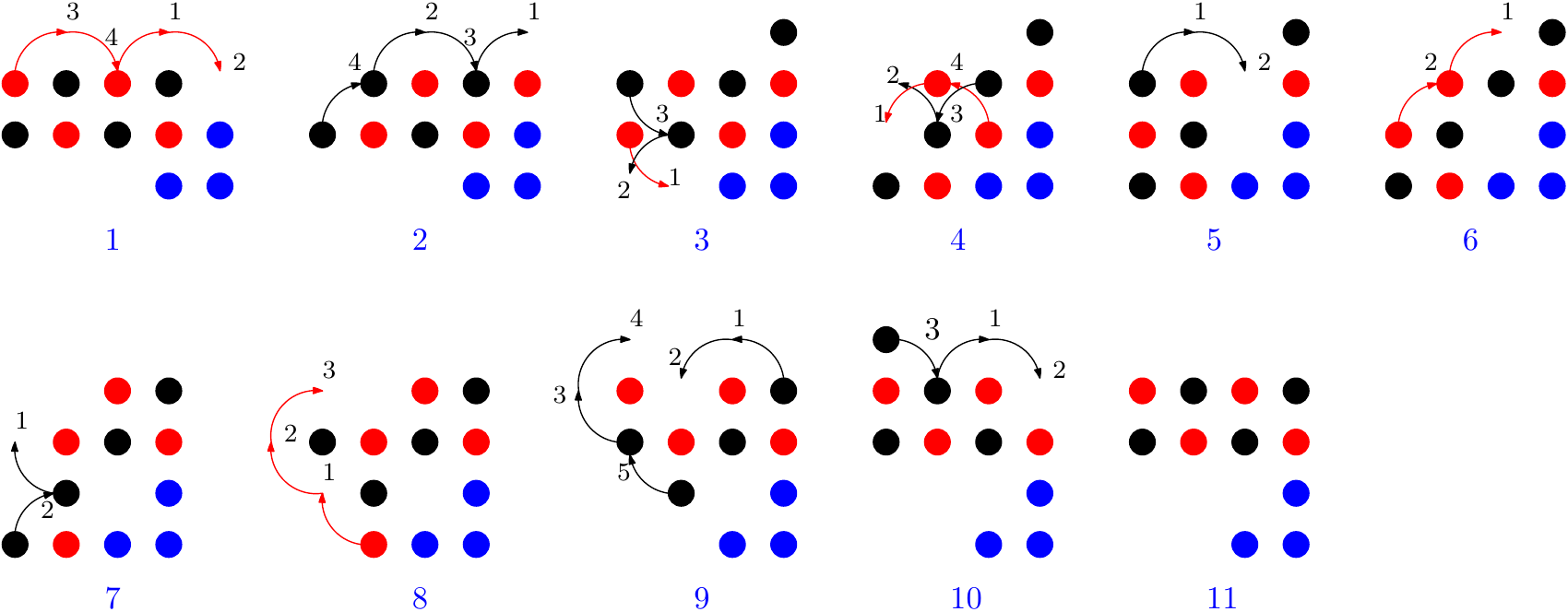}
}
\caption{Method for climbing a step.}\label{fig:8step1}
\end{figure}

\begin{figure}[!hbtp]
\centering{
\includegraphics[width=1.0\textwidth]{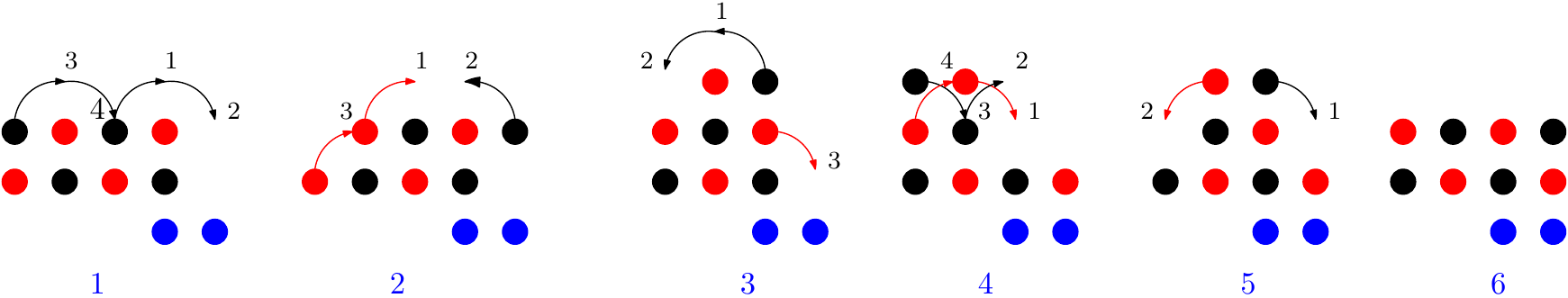}
}
\caption{Method for sliding.}\label{fig:8slide}
\end{figure}

\begin{figure}[!hbtp]
\centering{
\includegraphics[width=1.0\textwidth]{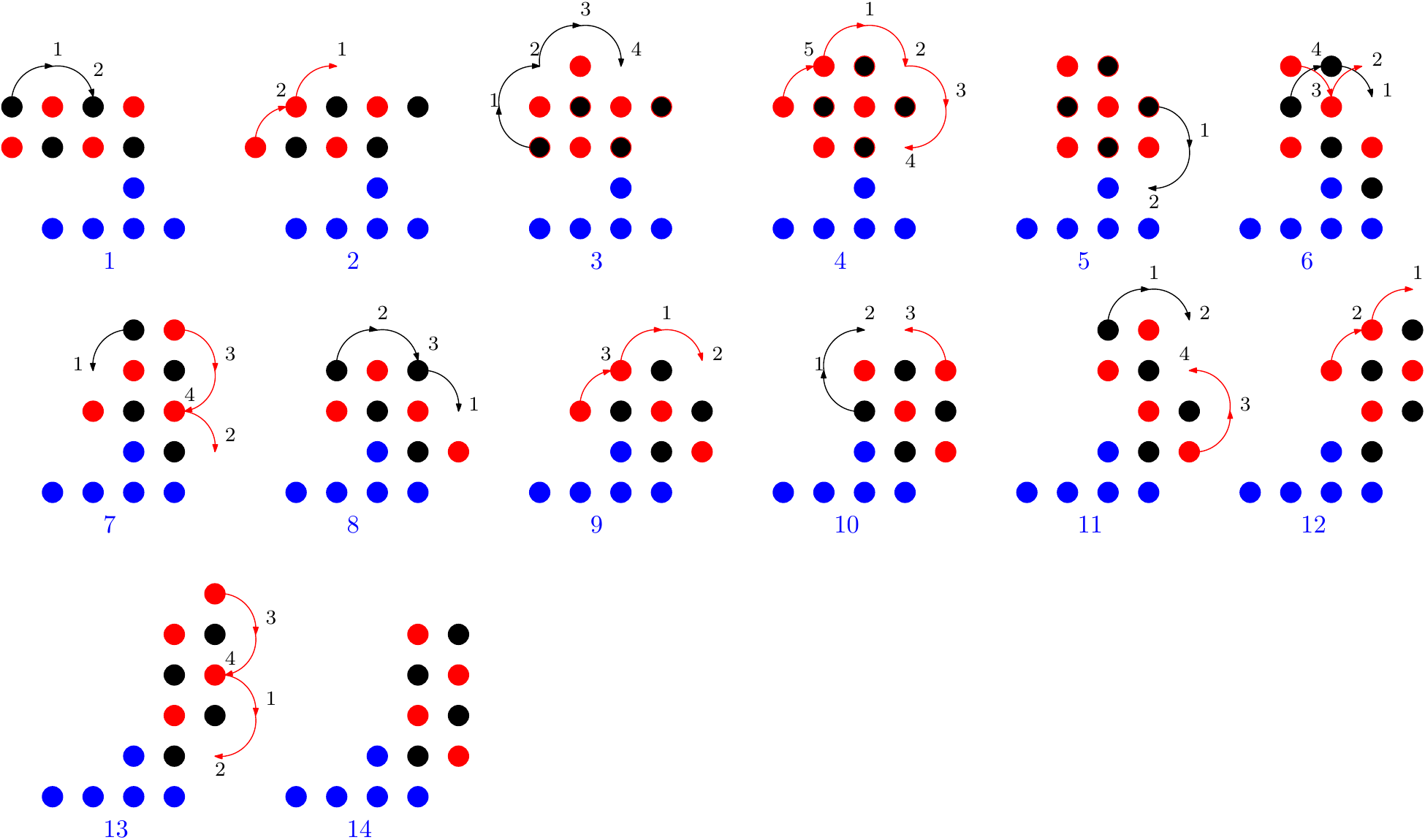}
}
\caption{Method for transforming.}\label{fig:8dstep}
\end{figure}

Our goal here was to show that since both a 6-seed and an 8-seed can traverse the perimeter of any discrete-convex shape, then a 6-seed may be able to start extracting $2$ nodes at a time from the shape A , move them as an 8-seed at a designated cell, leave them there, and continue this loop while creating i.e. a line with leaves. Afterwards we could perform the same method for shape B. If we succeeded in both shapes, then we could transform one to another.

\section{Rotation and Sliding}
\label{sec:rotation-sliding}

In this section, we study the combined effect of rotation and sliding movements.

We shall prove that rotation and sliding together, are \emph{transformation-universal}, meaning that they can transform any given shape to any other shape of the same size without ever breaking the connectivity during the transformation. It would be useful for the reader to recall Definitions \ref{def:perimeter1}, \ref{def:perimeter2}, \ref{def:perimeter3}, and \ref{def:perimeter4} and Proposition \ref{pro:extended-external-connected}, from Section \ref{sec:prel}, as the results that follow make extensive use of them.

As the perimeter is a (connected) polygon, it can be traversed by a particle walking on its edges (the unit-length segments). We now show how to ``simulate'' the particle's movement and traverse the cell-perimeter by a node, using rotation and sliding only.

\begin{lemma} \label{lem:walk-cell-perimeter-1}
If we place a node $u$ on any position of the cell-perimeter of a connected shape $A$, then $u$ can walk the whole cell-perimeter and return to its original position by using only rotations and slidings.
\end{lemma}
\begin{proof}
We show how to ``simulate'' the walk of a particle moving on the edges of the perimeter. The simulation implements the following simple rules:
\begin{enumerate}
\item If the current line-segment traversed by the particle concerns the same red cell as the one of the immediately previous line-segment traversed, then stay put.
\item If not:
 \begin{enumerate}
 \item If the two consecutive line-segments traversed form a line-segment of length 2, then move by sliding one position in the same direction as the particle.
 \item If the two consecutive line-segments traversed are perpendicular to each other, then move by a single rotation in the same direction as the particle.
 \end{enumerate}
\end{enumerate}

It remains to prove that $u$ can indeed always perform the claimed movements. (1) is trivial. For (2.a), a line-segment of length 2 on the perimeter is always defined by two consecutive blacks to the interior and two consecutive empty cells to the exterior (belonging to the cell-perimeter), therefore, $u$ can slide on the empty cells. For (2.b), there must be a black in the internal angle defined by the line-segments and an empty cell diagonally to it, in the exterior (for an example, see the right black node on the highest row containing nodes of $A$, in Figure \ref{fig:perimeter-definition}, Section \ref{sec:prel}). Therefore, rotation can be performed.
\end{proof}

Next, we shall prove that $u$ need not be an additional node, but actually a node belonging to the shape, and in particular one of those lying on the shape's boundary.

\begin{lemma} \label{lem:r-external-surface}
Let $A$ be a connected shape of order at least 2. Then there is a subset $R$ of the nodes on $A$'s external surface, such that $|R|\geq 2$ and for all $u\in R$, if we completely remove $u$ from $A$, then the resulting shape $A^\prime=A-\{u\}$ is also connected.
\end{lemma}
\begin{proof}
If the extended external surface of $A$ contains a cycle, then such a cycle must necessarily have length at least $4$ (due to geometry). In this case, any node of the intersection of the external surface (non-extended) and the cycle can be removed without breaking $A$'s connectivity. If the extended external surface of $A$ does not contain a cycle, then it corresponds to a tree graph which by definition has at least 2 leaves, i.e., nodes of degree exactly 1. Any such leaf can be removed without breaking $A$'s connectivity. In both cases, $|R|\geq 2$.
\end{proof}

\begin{lemma} \label{lem:walk-cell-perimeter-2}
Pick any $u\in R$ ($R$ defined on a connected shape $A$ as above). Then $u$ can walk the whole cell-perimeter of $A^\prime=A-\{u\}$ by rotations and slidings.
\end{lemma}
\begin{proof}
It suffices to observe that $u$ already lies on the cell-perimeter of $A^\prime$. Then, by Lemma \ref{lem:walk-cell-perimeter-1}, it follows that such a walk is possible.
\end{proof}

We are now ready to state and prove the universality theorem of rotations and slidings.

\begin{theorem} \label{the:universality-rot-sl}
Let $A$ and $B$ be any connected shapes, such that $|A|=|B|=n$. Then $A$ and $B$ can be transformed to each other by rotations and slidings, without breaking the connectivity during the transformation.
\end{theorem}
\begin{proof}
It suffices to show that any connected shape $A$ can be transformed to a spanning line $L$ by rotations and slidings only and without breaking connectivity during the transformation. If we show this, then $A$ can be transformed to $L$ and $B$ can be transformed to $L$ (as $A$ and $B$ have the same order, therefore correspond to the same spanning line $L$), and by reversibility of these movements, $A$ and $B$ can be transformed to each other via $L$.

Pick the rightmost column of the grid containing at least one node of $A$, and consider the lowest node of $A$ in that column. Call that node $u$. Observe that all cells to the right of $u$ are empty. Let the cell of $u$ be $(i,j)$. The final constructed line will start at $(i,j)$ and end at $(i,j+n-1)$.

The transformation is partitioned into $n-1$ phases. In each phase $k$, we pick a node from the original shape and move it to position $(i,j+k)$, that is, to the right of the right endpoint of the line formed so far. In phase 1, position $(i,j+1)$ is a cell of the cell-perimeter of $A$. So, even if it happens that $u$ is a node of degree 1, by Lemma \ref{lem:r-external-surface}, there must be another such node $v\in A$ that can walk the whole cell-perimeter of $A^\prime=A-\{v\}$ (the latter, due to Lemma \ref{lem:walk-cell-perimeter-2}). As $u\neq v$, $(i,j+1)$ is also part of the cell-perimeter of $A^\prime$, therefore, $v$ can move to $(i,j+1)$ by rotations and slidings. As $A^\prime$ is connected (by Lemma \ref{lem:r-external-surface}), $A^\prime\cup \{(i,j+1)\}$ is also connected and also all intermediate shapes were connected, because $v$ moved on the cell-perimeter and, therefore, it never disconnected from the rest of the shape during its movement.

In general, the transformation preserves the following invariant. At the beginning of phase $k$, $1\leq k\leq n-1$, there is a connected shape $S(k)$ (where $S(1)=A$) to the left of of column $j$ ($j$ inclusive) and a line of length $k-1$ starting from position $(i,j+1)$ and growing to the right. Restricting attention to $S(k)$, there is always a $v\neq u$ that could move to position $(i,j+1)$ if it were not occupied. This implies that before the final movement that places it on $(i,j+1)$, $v$ must have been in one of $(i+1,j)$ and $(i+1,j+1)$, if we assume that $v$ always walks in the clockwise direction. Observe now that from each of these positions $v$ can perform zero or more right slidings above the line in order to reach the position above the right endpoint of the line. When this occurs, a final clockwise rotation makes $v$ the new right endpoint of the line. The only exception is when $v$ is on $(i+1,j+1)$ and there is no line to the right of $(i,j)$ (this implies the existence of a node on $(i+1,j)$, otherwise connectivity of $S(k)$ would have been violated). In this case, $v$ just performs a single downward sliding to become the right endpoint of the line.
\end{proof}

\begin{theorem} \label{the:ladder-time}
The transformation of Theorem \ref{the:universality-rot-sl} requires $\Theta(n^2)$ movements in the worst case.
\end{theorem}
\begin{proof}
Consider a ladder shape of order $n$, as depicted in Figure \ref{fig:ladder}. The strategy of Theorem \ref{the:universality-rot-sl} will choose to construct the line to the right of node $u$. The only node that can be selected to move in each phase without breaking the shape's connectivity is the top-left node. Initially, this is $v$, which must perform $\lceil n/2\rceil$ movements to reach its position to the right of $u$. In general, the total number of movements $M$, performed by the transformation of Theorem \ref{the:universality-rot-sl} on the ladder, is given by
\begin{align*}
M&=\left\lceil \frac{n}{2}\right\rceil+2\cdot\sum_{i=1}^{(n-3)/2} \left\lceil \frac{n}{2}\right\rceil+i\\
&= \left\lceil \frac{n}{2}\right\rceil(n-2)+2\cdot\sum_{i=1}^{(n-3)/2} i\\
&= \Theta(n^2).
\end{align*}
\end{proof}

\begin{figure}[!hbtp]
\centering{
\includegraphics[width=0.8\textwidth]{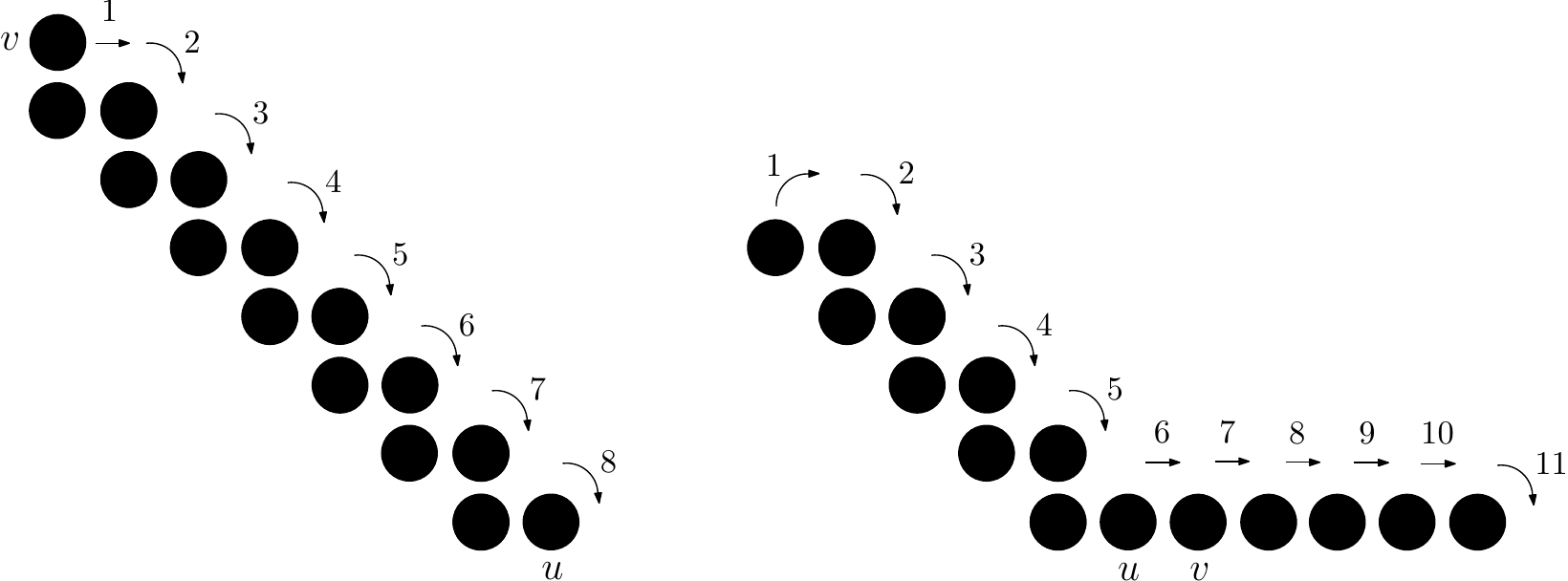}
}
\caption{Transforming a ladder into a spanning line.} \label{fig:ladder}
\end{figure}

Theorem \ref{the:ladder-time} shows that the above generic strategy is slow in some cases, as is the case of transforming a ladder shape into a spanning line. We shall now show that there are pairs of shapes for which any strategy and not only this particular one, may require a quadratic number of steps to transform one shape to the other.

\begin{definition} \label{def:potential-distance}
Define the \emph{potential of a shape} $A$ as its minimum ``distance'' from the line $L$, where $|A|=|L|$. The \emph{distance} is defined as follows: Consider any placement of $L$ relative to $A$ and any pairing of the nodes of $A$ to the nodes of the line. Then sum up the Manhattan distances \footnote{The Manhattan distance between two points $(i,j)$ and $(i^\prime,j^\prime)$ is given by $|i - i^\prime| + |j - j^\prime|$.} 
between the nodes of each pair. The minimum sum between all possible relative placements and all possible pairings is the distance between $A$ and $L$ and also $A$'s potential. In case the two shapes do not have an equal number of nodes, then any matching is not perfect and the distance can be defined as infinite.
\end{definition}

Observe that the potential of the line is 0 as it can be totally aligned on itself and the sum of the distances is 0.

\begin{lemma}
The potential of the ladder is $\Theta(n^2)$.
\end{lemma}
\begin{proof}
We prove it for horizontal placement of the line, as the vertical case is symmetric. Any such placement leaves either above or below it at least half of the nodes of the ladder (maybe minus 1). W.l.o.g. let it be above it. Every two nodes, the height increases by 1, therefore there are 2 nodes at distance 1, 2 at distance 2,$\ldots$, 2 at distance n/4. Any matching between these nodes and the nodes of the line gives for every pair a distance at least as large as the vertical distance between the ladder's node and the line, thus, the total distance is at least $2\cdot 1+2\cdot 2+...+2\cdot (n/4) = 2\cdot (1+2+...+n/4) = (n/4)\cdot(n/4 + 1) = \Theta(n^2)$. We conclude that the potential of the initial ladder is $\Theta(n^2)$.
\end{proof}

\begin{theorem}
Any transformation strategy based on rotations and slidings and performing a single movement per step, requires $\Theta(n^2)$ steps to transform a ladder into a line.
\end{theorem}
\begin{proof}
To show that $\Omega(n^2)$ movements are needed to convert the ladder to a line, it suffices to observe that the difference in their potentials is that much and that one rotation or one sliding can decrease the potential by at most 1.
\end{proof}

\begin{remark}
The above lower bound is independent of connectivity preservation. It is just a matter of the total distance based on single distance-one movements.
\end{remark}

Finally, it is interesting to observe that such lower bounds can be computed in polynomial time, because there is a polynomial-time algorithm for computing the distance between two shapes.

\begin{proposition}
Let $A$ and $B$ be connected shapes. Then their distance $d(A,B)$ can be computed in polynomial time.
\end{proposition}
\begin{proof}
The algorithm picks a node $u\in B$, a cell $c$ of the grid occupied by a node $v\in A$, and an orientation $o\in {north,east,south,west}$ and draws a copy of the shape $B$, starting with $u$ on $c$ and respecting the orientation $o$. Then, it constructs (in its memory) a complete weighted bipartite graph $(X,Y)$, where $X$ and $Y$ are equal to the node-sets of $A$ and $B$, respectively. The weight $w(x,y)$ for $x\in X$ and $y\in Y$ is defined as the distance from $x$ to $y$ (given the drawing of shape $B$ relative to shape $A$). To compute the minimum total distance pairing of the nodes of $A$ and $B$ for this particular placement of $A$ and $B$, the algorithm computes a minimum cost perfect matching of $(X,Y)$, e.g., by the Kuhn-Munkres algorithm (a.k.a. the Hungarian algorithm) \cite{Ku55}, and the sum of the weights of its edges $k$, and sets $dist=\min\{d,k\}$. Then the algorithm repeats for the next selection of $u\in B$, cell $c$ occupied by a node $v\in A$, and orientation $o$. In the end, the algorithm gives $dist$ as output. To see that $dist=d(A,B)$, observe that the algorithm just implements the procedure for computing the distance, of Definition \ref{def:potential-distance}, with the only differences being that it does not check all pairings of the nodes, instead directly computes the minimum-cost pairing, and that it does not try all relative placements of $A$ and $B$ but only those in which $A$ and $B$ share at least one cell of the grid. To see that this selection is w.l.o.g., assume that a placement of $A$ and $B$ in which no cell is shared achieves the minimum distance and observe that, in this case, $A$ could be shifted one step ``closer'' to $B$, strictly decreasing their distance and, thus, contradicting the optimality of such a placement. As the relative placements of $A$ and $B$ are $4n^2$ and the Kuhn-Munkres algorithm is a polynomial-time algorithm (in the size of the bipartite graph), we conclude that the algorithm computes the distance in polynomial time.
\end{proof}

To give a faster transformation either pipelining must be used (allowing for more than one movement in parallel) or more complex mechanisms that move sub-shapes consisting of many nodes, in a single step.

\subsection{Parallelizing the Transformations}

We now maintain the connectivity preservation requirement but allow an unbounded number of rotation and/or sliding movements to occur simultaneously in a single step.

\begin{proposition} \label{pro:ladder-pipelining}
There is a pipelining strategy that transforms a ladder into a line in $O(n)$ parallel time.
\end{proposition}
\begin{proof}
Number the nodes of the ladder 1 through $n$ starting from the top and following the ladder's connectivity until the bottom-right node is reached. These gives an even-numbered upper diagonal and an odd-numbered lower diagonal. Node 1 moves as in Theorem \ref{the:ladder-time}. Any even node $2\leq w <n-1$ starts moving as long as its upper odd neighbor has reached the same level as $w$ (e.g., node 2 first moves after node 1 has arrived to the right of node 3). Any odd node $1<z<n$ starts moving as long as its even left neighbor has moved one level down (e.g., node 3 first moves after node 2 has arrived to the right of 5). After a node starts moving, it moves in every step as in Theorem \ref{the:ladder-time} (but now many nodes can move in parallel, implementing a pipelining strategy). It can be immediately observed that any node $i$ starts after at most 3 movements of node $i-1$ (actually, only 2 movements for even $i$), so after roughly at most $3n$ steps, node $n-2$ starts. Moreover, a node that starts, arrives at the right endpoint of the line after at most $n$ steps, which means that after at most $4n=O(n)$ steps all nodes have taken their final position in the line.
\end{proof}

Proposition \ref{pro:ladder-pipelining} gives a hint that pipelining could be a general strategy to speed-up transformations. We next show how to generalize this technique to any possible pair of shapes.

\begin{theorem} \label{the:general-pipelining}
Let $A$ and $B$ be any connected shapes, such that $|A|=|B|=n$. Then there is a pipelining strategy that can transform $A$ to $B$ (and inversely) by rotations and slidings, without breaking the connectivity during the transformation, in $O(n)$ parallel time.
\end{theorem}
\begin{proof}
The transformation is a pipelined version of the sequential transformation of Theorem \ref{the:universality-rot-sl}. Now, instead of picking an arbitrary next candidate node of $S(k)$ to walk the cell-perimeter of $S(k)$ clockwise, we always pick the rightmost clockwise node $v_k\in S(k)$, that is, the node that has to walk the shortest clockwise distance to arrive at the line under formation. This implies that the subsequent candidate node $v_{k+1}$ to walk, is always ``behind'' $v_k$ in the clockwise direction and is either already free to move or is enabled after $v_k$'s departure. Observe that after at most 3 clockwise movements, $v_k$ cannot block any more the way of $v_{k+1}$ on the (possibly updated) cell-perimeter. Moreover, the clockwise move of $v_{k+1}$, only introduces a gap in its original position, therefore it only affects the structure of the cell-perimeter ``behind'' it. The strategy is to start the walk of node $v_{k+1}$ as soon as $v_k$ is no longer blocking its way. As in Proposition \ref{pro:ladder-pipelining}, once a node starts, it moves in every step, and again any node arrives after at most $n$ movements. It follows, that if the pipelined movement of nodes cannot be blocked in any way, after $4n=O(n)$ steps all nodes must have arrived at their final positions. Observe now that the only case in which pipelining could be blocked is when a node is sliding through a (necessarily dead-end) ``tunnel'' of height 1 (such an example is the red tunnel on the third row from the bottom, in Figure \ref{fig:perimeter-definition}). To avoid this, the nodes shortcut the tunnel by visiting only its first position $(i,j)$ and then simply skipping the whole walk inside it (that walk would just return them to position $(i,j)$ after a number of steps).
\end{proof}

We next show that even if $A$ and $B$ are labeled shapes, that is, their nodes are assigned the indices $1,\ldots,n$ (uniquely, i.e., without repetitions), we can still transform the labeled $A$ to the labeled $B$ with only a linear increase in parallel time. We only consider transformations in which the nodes never change indices in any way (e.g., cannot transfer them, or swap them), so that each particular node of $A$ must eventually occupy (physically) a particular position of $B$ (the one corresponding to its index).

\begin{corollary}
The labeled version of the transformation of Theorem \ref{the:general-pipelining} can be performed in $O(n)$ parallel time.
\end{corollary}
\begin{proof}
Recall from Theorem \ref{the:universality-rot-sl} that the line were constructed to the right of some node $u$. That node was the lowest node in that column, therefore, there is no node below $u$ in that column. The procedure of Theorem \ref{the:general-pipelining}, if applied on the labeled versions of $A$ and $B$ will result in two (possibly differently) labeled lines, corresponding to two permutations of $1,2,\ldots,n$, call them $\pi_A$ and $\pi_B$. It suffices to show a way to transform $\pi_A$ to $\pi_B$ in linear parallel time, as then labelled $A$ is transformed to $\pi_A$, then $\pi_A$ to $\pi_B$, and then $\pi_B$ to $B$ (by reversing the transformation from $B$ to $\pi_B$), all in linear parallel time.

To do this, we actually slightly modify the procedure of Theorem \ref{the:general-pipelining}, so that it does not construct $\pi_A$ in the form of a line, but in a different form that will allow us to transform it fast to $\pi_B$ without breaking connectivity. What we will construct is a double line, with the upper part growing to the right of node $u$ as before and the lower part starting from the position just below $u$ and also growing to the right. The upper line is an unordered version of the left half of $\pi_B$ and the lower line is an unordered version of the right half of $\pi_B$. To implement the modification, when a node arrives above $u$, as before, if it belongs to the upper line, it goes to the right endpoint of the line as before, while if it belongs to the lower line, it continues its walk in order to teach the right endpoint of the lower line.

When the transformation of labeled $A$ to the folded line is over, the procedure has to order the nodes of the folded line and then unfold in order to produce $\pi_B$. We first order the upper line in ascending order. While we do this, the lower line stays still in order to preserve the connectivity. When we are done, we order the lower line in descending order, now keeping the upper line still. Finally, we perform a parallel right sliding of the lower line (requiring linear parallel time), so that its inverse permutation ends up to the right of the upper line, thus forming $\pi$.

It remains to show how the ordering of the upper line can be done in linear parallel time without breaking connectivity. To do this, we simulate a version of the odd-even sort algorithm (a.k.a. parallel bubble sort) which sorts a list of $n$ numbers with $O(n)$ processors in $O(n)$ parallel time. The algorithm progresses in odd and even phases. In the odd phases, the odd positions are compared to their right neighbor and in the even phases to their left neighbor and if two neighbors are ever found not to respect the ordering a swap of their values is performed. In our simulation, we break each phase into two subphases as follows. Instead of performing all comparisons at once, as we cannot do this and preserve connectivity, in the first subphase we do every second of them and in the second subphase the rest so that between any pair of nodes being compared there are 2 nodes that are not being compared at the same time. Now if the comparison between the $ith$ and the $i+1$ node indicates a swap, then $i+1$ rotates over $i+2$, $i$ slides right to occupy the previous position of $i+1$, and finally $i+1$ slides left over $i$ and then rotates left around $i$ to occupy $i$'s previous position. This swapping need 4 steps and does not break connectivity. The upper part has $n/2$ nodes, each subphase takes 4 steps to swap everyone (in parallel), each phase has 2 sub-phases, and $O(n)$ phases are required for the ordering to complete, therefore, the total parallel time is $O(n)$ for the upper part and similarly $O(n)$ for the lower part. This completes the proof.
\end{proof}

An immediate observation is that a linear-time transformation does not seem satisfactory for all pairs of shapes. To this end, take a square $S$ and rotate its top-left corner $u$, one position clockwise, to obtain an almost-square $S^\prime$. Even though, a single counter-clockwise rotation of $u$ suffices to transform $S^\prime$ to $S$, the transformation of Theorem \ref{the:general-pipelining} may go all the way around and first transform $S^\prime$ to a line and then transform the line to $S$. In this particular example, the distance between $S$ and $S^\prime$, according to Definition \ref{def:potential-distance}, is $2$, while the generic transformation requires $\Theta(n)$ parallel time. So, it is plausible to ask if any transformation between two shapes $A$ and $B$ can be performed in time that grows as a function of their distance $d(A,B)$. We show that this cannot always be the case, by presenting two shapes $A$ and $B$ with $d(A,B)=2$, such that $A$ and $B$ require $\Omega(n)$ parallel time to be transformed to each other.

\begin{proposition}
There are two shapes $A$ and $B$ with $d(A,B)=2$, such that $A$ and $B$ require $\Omega(n)$ parallel time to be transformed to each other.
\end{proposition}
\begin{proof}
The two shapes, a black and a red one, are depicted in Figure \ref{fig:distance-counterexample}. Both shapes form a square which is empty inside and also open close to the middle of its bottom side. The difference between the two shapes is the positioning of the bottom ``door'' of length 2. The red shape has it exactly in the middle of the side, while the black shape has it shifted one position to the left. Equivalently, the bottom side of the red shape is ``balanced'', meaning that it has an equal number of nodes in each side of the vertical dashed axis that passes through the middle of the bottom, while the black shape is ``unbalanced'' having one more node to the right of the vertical axis than to its left.

To transform the black shape into the red one, a node must necessary cross either the vertical or the horizontal axis. Because, if nothing of the two happens, then, no matter the transformation, we won't be able to place the axes so that the running shape has two pairs of balanced quadrants, while, on the other hand, the red shape satisfies this, by pairing together the two bottom quadrants and the two upper quadrants. Clearly, no move can be performed in the upper quadrants initially, as this would break the shape's connectivity. The only nodes that can move initially are $u$ and $w$ and no other node can ever move unless first approached by some other node that could already move. Observe also that $u$ and $w$ cannot cross the vertical boundary of their quadrants, unless with help of other nodes. But the only way for a second node to move in any of these quadrants (without breaking connectivity) is for either $u$ or $w$ to reach the corner of their quadrant which takes at least $n/8-2$ steps and then another $n/8$ steps for any (or both) of these nodes to reach the boundary, that is, at least $n/4-2$ steps, which already proves the required $\Omega(n)$ parallel-time lower bound (even a parallel algorithm has to pay the initial sequential movement of either $u$ or $w$).
\end{proof}

\begin{figure}[!hbtp]
\centering{
\includegraphics[width=0.35\textwidth]{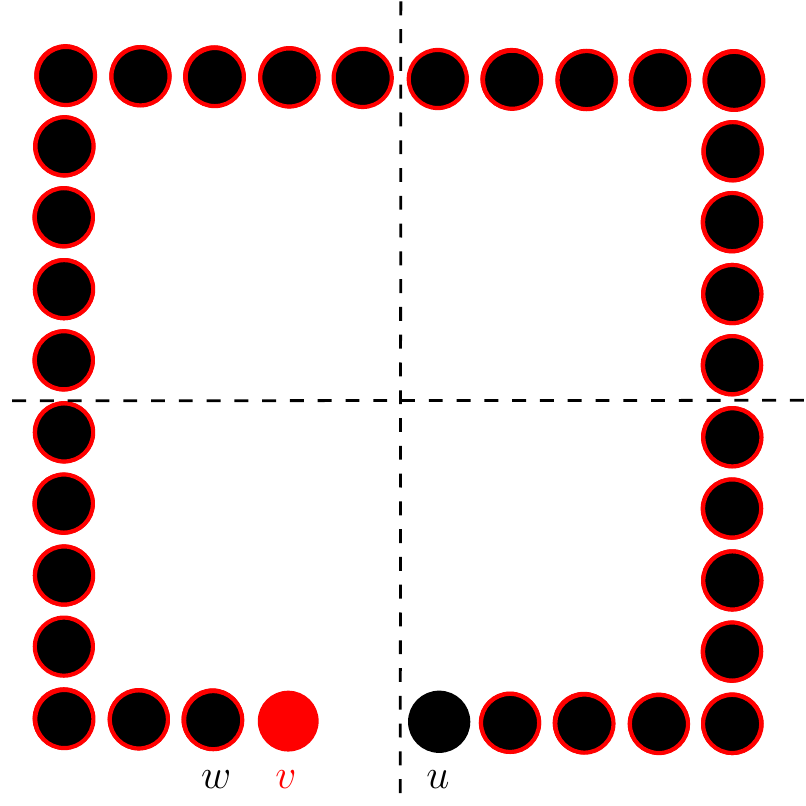}
}
\caption{Counterexample for distance} \label{fig:distance-counterexample}
\end{figure}

\section{Distributed Transformations with Rotation and Sliding}
\label{sec:distributed}

$\rem{Line Transformation Problem}$: Transform any connected shape into a shape where all nodes are either in the $x$ axis or the $y$ axis.

\begin{definition}
A node on a black cell of the grid is called a potential hole node if removing him would create a non-compact shape.
\end{definition}

\begin{definition}
We define a \emph{local-info based} movement (lib movement), a movement that a leader decides to perform without consulting the whole network.
\end{definition}

\begin{proposition}
No algorithm based on lib movements can solve the line transformation problem without breaking the connectivity.
\end{proposition}

\begin{proof}
Observe the following shapes. See figure \ref{fig:lib-movement-counterexample} If an algorithm performed a lib movement at the first shape, the same algorithm would have to perform the same movement at the second shape because it cannot distinguish the two shapes. That movement would break the connectivity on the second shape therefore no algorithm based on lib movements could solve the line formation problem on both shapes.
\end{proof}

\begin{figure}[!hbtp]
\centering{
\includegraphics[width=0.35\textwidth]{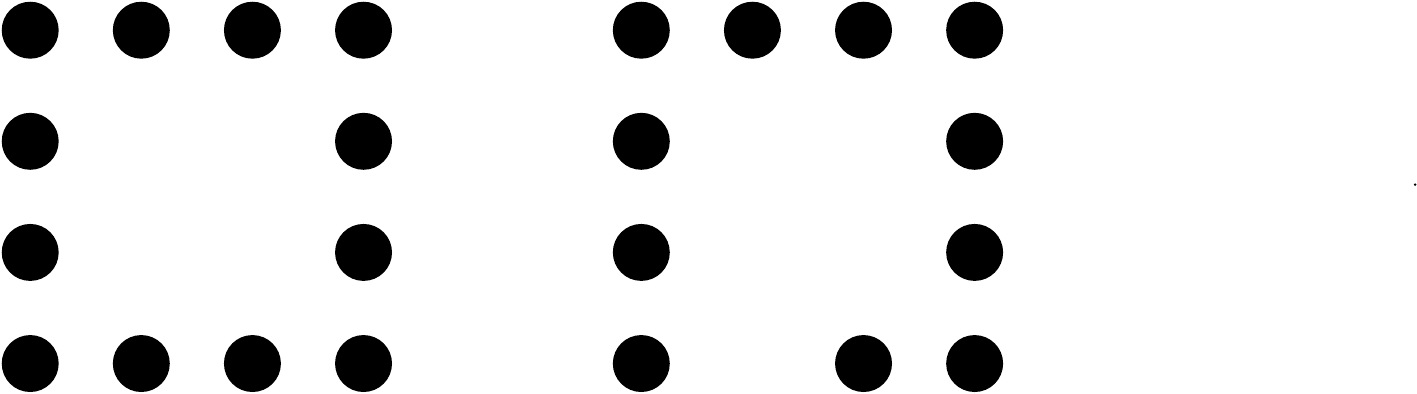}
}
\caption{Lib movement counterexample} \label{fig:lib-movement-counterexample}
\end{figure}

\begin{definition}
A shape is called compact when it has no holes.
\end{definition}

\begin{lemma}
If all nodes in a connected shape have two or more neighbours, then there it at least one cycle present.
\end{lemma}

\begin{proof}
Assume a connected graph where every node has at least $2$ neighbours and there is no cycle present. Every node has $2$ edges and we have $n$ nodes. The sum of all edges is $2n/2=n$ because we have double counted every edge. Now a connected graph without a cycle is called a tree and a tree has $n-1$ edges. If and edge is added (n edges) the tree creates a cycle.
\end{proof}

\begin{lemma}
In compact shapes, there is always a lib movement that does not break the connectivity.
\end{lemma}

\begin{proof}
Observe that there is no shape where all nodes have $3$ or $4$ neighbours and when a node has only one neighbour it can always perform a lib movement. Suppose that a shape exists where no lib movements are available. Each node in this shape has at least $2$ neighbours so there is at least one cycle. If this cycle is compact, then a node can consult his neighbours to see if it is a bridge in order to move. (HeIt can do this by asking whether the two neighbours it has, have a different node as a common neighbor). This is a lib movement because the information required is local based. Therefore a lib movement is always available.
\end{proof}

\begin{algorithm}
\begin{algorithmic}[1]
\Procedure{Compact-Line}{Leader} \Comment{The difference between a move and a travel action is, that the leader moving means that it swaps his settings to a neighbouring node, while a leader traveling means that it moves his current node.}
\State $label\gets 0$, $phase\gets 0$, $check\gets 0$, $state\gets 0$, $buck\gets 0$, $num\gets 0$, $line\gets 0$
\While{$r=1$}
    \State $p_0\gets north$, $p_1\gets east$, $p_2\gets south$, $p_3\gets west$, $orien(node)=1;$
    \State $Send$ $north$, $east$, $south$, $west$ to $p_0$, $p_1$, $p_2$, $p_3$ respectively$;$
\EndWhile
\While{$r=2,3$ AND $buck=0$}
    \State $Send$ $tick=up, right, down, left$; $Send$ $num=0,1,0,-1$ to $p_0$, $p_1$, $p_2$, $p_3$ respectively$;$
    \State $Receive$ $(tick',num',ack');$
    \State $Do$ for every $tick',num'$ received$;$
        \If{$num'>num$}
            \State $line=tick';$ $num=num';$
        \EndIf
    \If{$ack'=null$}
        \State $state++;$
    \Else
        \State $state--;$
    \EndIf
    \If{$state\ge2$}
        \State $follow$ the path described on the line$;$
    \EndIf
    \If{line has been reached}
        \State $buck=1;$
    \EndIf
\EndWhile
\While{$r=2,3,..$ AND $buck=1$}
    \State $Receive(flag',qu');$
    \If {$label=0$}
        \State $move$ $west;$ $flag(node)=1;$
        \If{no node available west}
            \State $label++;$
        \EndIf
    \EndIf
    \If {$label=1$}
        \State $move$ $east;$  $flag(node)=1;$
        \If{no node available east}
            \State $label++;$
        \EndIf
    \EndIf
    \If{$flag'=0$ AND $flag=0$}
        \If{$label=2$}
            \If{$neighbours=1$}
                \State $label++;$
            \ElsIf{$neighbours=2$ AND $neighbours$ not opposite}
                \If{$neighbours$ are $p_0,p_1$}
                    \State $send$ $qu=1$  to $p_0;$ $label++;$
                \EndIf
                \If{$neighbours$ are $p_1,p_2$}
                    \State $send$ $qu=1$  to $p_1$; $label++;$
                \EndIf
        \algstore{myalg1}
\end{algorithmic}
\end{algorithm}

\begin{algorithm}
\begin{algorithmic}
\algrestore{myalg1}
                \If{$neighbours$ are $p_2,p_3$}
                    \State $send$ $qu=1$  to $p_2;$ $label++;$
                \EndIf
                \If{$neighbours$ are $p_3,p_0$}
                    \State $send$ $qu=1$  to $p_3;$  $label++;$
                \EndIf
            \Else
                \State $move;$
            \EndIf
        \EndIf
\If{$label=3$ AND $phase=0$}
            \If{$qu'=1$}
                \State $phase=1;$
            \Else
                \State $label--;$ $move;$
            \EndIf
        \EndIf
        \If{$label=3$ AND $phase=1$}
            \State $travel;$
        \EndIf
    \EndIf
    \If{$flag'=1$ OR $flag=1$}
        \State $flag=1;$ $phase=0;$ $label=2;$ $move;$
    \EndIf
\EndWhile
\EndProcedure
\end{algorithmic}
\end{algorithm}

\begin{algorithm}
\begin{algorithmic}[1]
\Procedure{Compact-Line}{Non-Leader}
\State $flag\gets 0$, $orien\gets 0$, $mark\gets 0$
\While{$orien=0$}
    \State $Receive(north', south', east', west', tick', num', qu')$
    \If{$p_0$ receives $south$ OR $p_1$ receives $west$ OR $p_2$ receives $north$ OR $p_3$ receives $east$}
    \State $orien=1;$
    \EndIf
    \If{$p_0$ receives $east$ OR $p_1$ receives $south$ OR $p_2$ receives $west$ OR $p_3$ receives $north$}
    \State $p_n=p_{(n-1)mod3};$ $orien=1;$
    \EndIf
    \If{$p_0$ receives $north$ OR $p_1$ receives $east$ OR $p_2$ receives $south$ OR $p_3$ receives $swest$}
    \State $p_n=p_{(n-2)mod3};$ $orien=1;$
    \EndIf
    \If{$p_0$ receives $west$ OR $p_1$ receives $north$ OR $p_2$ receives $east$ OR $p_3$ receives $south$}
    \State $p_n=p_{(n-3)mod3};$ $orien=1;$
    \EndIf
    \State $Send north, east, south, west$ to $p_0, p_1, p_2, p_3$ respectively$;$
\EndWhile
\While{$orien=1$}
    \If{$mark=0$}
        \State $Receive(tick',num');$
        \State $Send$ $tick=$conc$(tick',up)$, $tick=$conc$(tick',right)$, $tick=$conc$(tick',down)$, $tick=$conc$(tick',left)$ to $p_0, p_1, p_2 p_3$ respectively$;$
        \State $Send$ $num=num'+0$, $num=num'+1$, $num=num'+0$, $num=num'-1$ to $p_0, p_1, p_2, p_3$ respectively$;$
        \State $Send$ $ack=1$ to the node who sent you $tick$;
        \State $mark=1;$ $path=m$ where m is the number of the port that received the $tick'-num';$ message$;$
    \EndIf
    \algstore{myalg2}
\end{algorithmic}
\end{algorithm}

\begin{algorithm}
\begin{algorithmic}
\algrestore{myalg2}
    \If{mark=1}
        \State $Receive(ack',tick',qu');$
        \If{$tick'$ was not null}
            \State $Send (ack=1)$ to $p_{path};$
        \ElsIf{$ack'$ was not null}
            \State $Send (ack=1)$ to $p_{path};$
        \EndIf
        \If{$qu=0$}
            \If{$p_0$ has neighbour}
                \State $Send qu=approve;$
            \Else
                \State $Send qu=reject;$
            \EndIf
        \EndIf
        \If{$qu=1$}
            \If{$p_1$ has neighbour}
                \State $Send qu=approve;$
            \Else
                \State $Send qu=reject;$
            \EndIf
        \EndIf
        \If{$qu=2$}
            \If{$p_2$ has neighbour}
                \State $Send qu=approve;$
            \Else
                \State $Send qu=reject;$
            \EndIf
        \EndIf
        \If{$qu=3$}
            \If{$p_3$ has neighbour}
                \State $Send qu=approve;$
            \Else
                \State $Send qu=reject;$
            \EndIf
        \EndIf
        \State $Send flag$ to $p_1;$
    \EndIf
\EndWhile
\EndProcedure
\end{algorithmic}
\end{algorithm}

\begin{theorem}
The Compact Line algorithm solves the Line Formation problem for any compact starting shape without breaking the connectivity.
\end{theorem}

\emph{Algorithm Description}: The operation of the algorithm is split into $3$ stages. The orientation stage$(1)$, the line marking stage$(2)$ and the movement stage$(3)$.
The first stage consists of the leader starting from a random node. It sets the orientation for the current node by marking ports $0$,$1$,$2$,$3$ as ``north'', ``east'', ``south'', ``west'' respectively. It then sends the orientation to all neighbours. All nodes receiving the orientation change their ports to coincide with the one the leader defined, and then propagate the message to their neighbours. Once a node sets its orientation once, it ignores all messages that concern it.

In the second stage the leader searches for the rightmost node. It begins by broadcasting two messages to all neighbours: \emph{tick} and \emph{num}. The tick message consists of the direction the message was sent to. The \emph{num} message is a number which starts as $0$ and each time it is propagated through nodes, we add the following number: $0$ for \emph{north} , $+1$ for \emph{east}, $0$ for \emph{south} and $-1$ for \emph{west}. When a non leader node receives these messages, it propagates them to its neighbours after appending \emph{up}, \emph{right}, \emph{down}, \emph{left}, for neighbours $0$,$1$,$2$,$3$ respectively, to the \emph{tick} message and after adding the number to the \emph{num} message following the method mentioned above. The node also sends a message called \emph{ack} to the node who sent the \emph{tick} and \emph{num}. It then stores the node (\emph{path node}) who sent the \emph{tick} and \emph{marks} himself. When a marked node receives a \emph{tick-num} message it sends them to the \emph{path node} along with an \emph{ack} message. When the leader receives a \emph{num-tick} message, it compares the \emph{num} it received with the \emph{num$^\prime$} it has in store. If the one it received is bigger, it replaces the \emph{num$^\prime$} with \emph{num} and keeps the \emph{tick$^\prime$} message it received. Now, if the leader does not receive an \emph{ack} for two consecutive rounds it starts following the path it has stored in the variable named \emph{line}. Once it reaches the destination it marks the current node and starts moving \emph{west}, marking all nodes in its path. It then returns to the node it marked first. The leader has now marked a designated line where it will move all other nodes to. This ends phase $2$.

The third stage consists of a loop being performed until all nodes form a line. Loop description: The leader moves randomly to nodes checking if they are on the correct line ($flag=1$). If it finds one and receives a message (flag$^\prime$=1), the leader marks it. If it finds one does not receive a message (flag=1), it checks two things. First it checks if the node has only one neighbour. Secondly it checks if the node has two neighbours not opposite to each other. If it does complete the second requirement, it sends a \emph{qu} message to one of them asking it if the $2$ nodes who are neighbours to it (the leader), have another common neighbour. The node then answers \emph{approve} or \emph{reject}. If any of those two checks are true (one neighbour, approve) the leader moves in a random fashion. Once it receives a message \emph{flag$^\prime$=1}, it marks the node. That finishes the loop description.

\begin{proof}
The goal is to show that any compact shape will always end up in a line. We need to show that the connectivity will be preserved throughout the transformation and the shape will not get stuck in a shape which is not a line.

The first two phases cannot break the connectivity because no movement takes place. The leader is moving between nodes. The third phase consists of the loop. The loop has three phases. The lib movement, the search for the line, and the placement on the line. Observe that if a node performs a lib movement and iy does not break the connectivity, all its subsequent moves will not break it as long as no other node has moved. This guarantees that if the lib movement preserves the connectivity, the second and third phase will preserve it as well. Now in both checks performed by the algorithm, it is ensured that the node has no bridges so any movement it performs will not break the connectivity.

The first two phases will not turn the shape into an undesired one, because no movement takes place. All we have to show is that the loop won't do it as well.
The first case is whether a lib movement will always be available. Observe that there is no shape where all nodes have $3$ or $4$ neighbours. In this protocol we have a lib movement available when a node has only $1$ neighbour or when a cluster of $4$ nodes creates a square. Suppose that there is a shape that has no lib movements available. Every node has at least $2$ neighbours. Such a shape can only be a cycle consisting of all nodes. This cycle is either a non compact shape or it consists of only $4$ nodes. But both of those circumstances are prohibited because we have established that we are talking about compact shapes where there are no lib movements. Thus a lib movement is always available.

The second case is whether the moving node will always find the right point of rightmost node of the line. Since the moving node is moving in random it will not get stuck in a loop. The rightmost node of the line is always accessible from any node on the shape.

The third problem is whether the move on the rightmost part of the line creates a non compact shape. Since the line begins from the rightmost part of the shape, each time a node moves east of it, it has only one neighbour which is the rest of the line. So it cannot create a non-compact shape.
\end{proof}

\begin{theorem}
The Line Transform algorithm solves the Line Formation problem for any starting shape without breaking the connectivity.
\end{theorem}

\begin{algorithm}
\begin{algorithmic}[1]
\Procedure{Line-Transform}{Leader}
\State $label\gets 0$, $phase\gets 0$, $check\gets 0$, $state\gets 0$, $buck\gets 0$, $num\gets 0$, $line\gets 0$
\While{$r=1$}
    \State $p_0\gets north$, $p_1\gets east$, $p_2\gets south$, $p_3\gets west$, $orien(node)=1;$
    \State $Send$ $north$, $east$, $south$, $west$ to $p_0$, $p_1$, $p_2$, $p_3$ respectively$;$
\EndWhile
\While{$r=2,3$ AND $buck=0$}
    \State $Send$ $tick=up, right, down, left$ $and$ $num=0,1,0,-1$ to $p_0$, $p_1$, $p_2$, $p_3$ respectively$;$
    \State $Receive$ $(tick',num',ack');$
    \State $Do$ for every $tick',num'$ received$;$
        \If{$num'>num$}
            \State $line=tick';$ $num=num';$
        \EndIf
    \If{$ack'=null$}
        \State $state++;$
    \Else
        \State $state--;$
    \EndIf
    \If{$state\ge2$}
        \State $follow$ the path described on the line$;$
    \EndIf
    \If{line has been reached}
        \State $buck=1;$
    \EndIf
\EndWhile
\While{$r=2,3,..$ AND $buck=1$}
    \State $Receive(flag',qu');$
    \If {$label=0$}
        \State $move$ $west;$ $flag(node)=1;$
        \If{no node available west}
            \State $label++;$
        \EndIf
    \EndIf
    \If {$label=1$}
        \State $move$ $east;$  $flag(node)=1$
        \If{no node available east}
            \State $label++;$
        \EndIf
    \EndIf
    \If($flag'=0$ AND $flag=0$)
        \If{$label=2$}
            \If{$neighbours=1$}
                \State $label++;$
            \ElsIf{$neighbours=2$ AND $neighbours$ not opposite}
                \If{$neighbours$ are $p_0,p_1$}
                    \State $send$ $qu=1$  to $p_0;$ $label++;$
                \EndIf
\algstore{myalg3}
\end{algorithmic}
\end{algorithm}

\begin{algorithm}
\begin{algorithmic}
\algrestore{myalg3}
                \If{$neighbours$ are $p_1,p_2$}
                    \State $send$ $qu=1$  to $p_1$; $label++;$
                \EndIf
                \If{$neighbours$ are $p_2,p_3$}
                    \State $send$ $qu=1$  to $p_2;$ $label++;$
                \EndIf
                \If{$neighbours$ are $p_3,p_0$}
                    \State $send$ $qu=1$  to $p_3;$  $label++;$
                \EndIf
            \ElsIf{$neighbour=2$ AND $neighbours opposite$}
                \If{neighbours are $p_0, p_2$}
                    \State $Send qu=2$ to $p_0$; $label++$;
                \EndIf
                \If{neighbours are $p_1, p_3$}
                    \State $Send qu=3$ to $p_1$; $label++$;
                \EndIf
            \Else
                \State $move;$
            \EndIf
        \EndIf
        \If{$label=3$ AND $phase=0$}
            \If{$ack'$ is null}
                \State $wait++;$
            \EndIf
            \If{$ack'$ is not null}
                \State $wait--;$
            \EndIf
            \If{$wait=2$}
                \State $label--;$ $move;$
            \EndIf
            \If{$qu'=k$ where k is the number of the port from where the qu' was received}
                \State $phase=1;$ $Send clear=1$ to all ports$;$
            \EndIf
        \EndIf
        \If{$label=3$ AND $phase=1$}
            \State $travel;$
        \EndIf
    \EndIf
    \If{$flag'=1$ OR $flag=1$}
        \State $flag=1;$ $phase=0;$ $label=2;$ $travel;$
    \EndIf
\EndWhile
\EndProcedure
\end{algorithmic}
\end{algorithm}

\begin{algorithm}
\begin{algorithmic}[1]
\Procedure{Line-Transform}{Non-Leader}
\State $flag\gets 0$, $orien\gets 0$, $mark\gets 0$
\While{$orien=0$}
    \State $Receive(north', south', east', west', tick', num', qu')$
    \If{$p_0$ receives $south$ OR $p_1$ receives $west$ OR $p_2$ receives $north$ OR $p_3$ receives $east$}
    \State $orien=1;$
    \EndIf
    \If{$p_0$ receives $east$ OR $p_1$ receives $south$ OR $p_2$ receives $west$ OR $p_3$ receives $north$}
    \State $p_n=p_{(n-1)mod3};$ $orien=1;$
    \EndIf
    \If{$p_0$ receives $north$ OR $p_1$ receives $east$ OR $p_2$ receives $south$ OR $p_3$ receives $swest$}
    \State $p_n=p_{(n-2)mod3};$ $orien=1;$
    \EndIf
    \If{$p_0$ receives $west$ OR $p_1$ receives $north$ OR $p_2$ receives $east$ OR $p_3$ receives $south$}
    \State $p_n=p_{(n-3)mod3};$ $orien=1;$
    \EndIf
    \State $Send north, east, south, west$ to $p_0, p_1, p_2, p_3$ respectively$;$
\EndWhile
\While{$orien=1$}
    \If{$mark=0$}
        \State $Receive(tick',num'); Send$ $tick=$conc$(tick',up)$, $tick=$conc$(tick',right)$, $tick=$conc$(tick',down)$, $tick=$conc$(tick',left)$ to $p_0, p_1, p_2 p_3$ respectively$;$
        \State $Send$ $num=num'+0$, $num=num'+1$, $num=num'+0$, $num=num'-1$ to $p_0, p_1, p_2, p_3$ respectively$;$
        \State $Send$ $ack=1$ to the node who sent you $tick$
        \State $mark=1;$ $path=m$ where m is the number of the port that received the $tick'-num';$ message$;$
    \EndIf
    \If{mark=1}
        \State $Receive(ack',tick',qu');$
        \If{$tick'$ was not null}
            \State $Send (ack=1)$ to $p_{path};$
        \ElsIf{$ack'$ was not null}
            \State $Send (ack=1)$ to $p_{path};$
        \EndIf
        \If{$qu'=0,1,2,3$}
            \State $path=m$ where m is the number of the port that received the $qu';$ message$;$
            \State $Send qu=qu'$ to all ports$;$ $Send ack=1$ to $port_{path};$
            $mark=2;$
        \EndIf
        \State $Send$ $flag$ to $p_1;$
    \EndIf
    \algstore{myalg4}

\end{algorithmic}
\end{algorithm}

\begin{algorithm}
\begin{algorithmic}
\algrestore{myalg4}
\If{$mark=2$}
        \State $Send flag east; Receive(ack',tick',clear');$
        \If{$ack'$ was not blank}
            \State $Send$ $ack=1$ to $p_{path}$
        \EndIf
        \If{$clear'=1$ from $p_{path}$}
            \State$Send clear=1$ to all ports$;$ $path=null;$ $mark=1;$
        \EndIf
    \EndIf
\EndWhile
\EndProcedure
\end{algorithmic}
\end{algorithm}

\section{Conclusions and Further Research}
\label{sec:conclusions}

There are many open problems related to the findings of the present work. We here restricted attention to the two extremes, in which the transformation either preserves connectivity or is free to break it arbitrarily. A compromise could be to allow some restricted degree of connectivity breaking, like necessarily restoring it in at most $k\geq 0$ steps (a special case of this had been already proposed as an open question in \cite{DP04}). There are other meaningful ``good'' properties that we would like to maintain throughout a transformation. An interesting example, is the \emph{strength} of the shape. One of the various plausible definitions is as the minimum strength sub-shape of the shape (i.e., its weakest part; could possibly be captured by some sort of minimum geometric cuts). Then, a strength-preserving transformation would be one that reaches the target shape while trying to maximize this minimum.

In the transformations considered in this paper, there was no \emph{a priori} constraint on the maximum area that a transformation is allowed to cover or on the maximum dimensions that its intermediate shapes are allowed to have. It seems in general harder to achieve a particular transformation if any of these restrictions is imposed. For example, the generic transformation of Theorem \ref{the:rotation-generic} requires some additional space below the shape and the transformations of Theorems \ref{the:universality-rot-sl} and \ref{the:general-pipelining} convert any shape first to a spanning line, whose maximum dimension is $n$, even though the original shape could have a maximum dimension as small as $\sqrt{n}$. Another interesting fact about restricting the boundaries is that in this way we get models equivalent to several interesting puzzles. For example, if the nodes are labeled, the initial shape is a square with a single empty cell, and the boundaries are restricted to the dimensions of the square, we get a generalization of the famous 15-puzzle (see, e.g., \cite{De01} for a very nice exposition of this and many more puzzles and 2-player games). Techniques developed in the context of puzzles could prove valuable for analyzing and characterizing discrete programmable matter systems.

We intentionally restricted attention to very minimal actuation mechanisms, namely rotation and sliding. More sophisticated mechanical operations would enable a larger set of transformations and possibly also reduce the time complexity. Such examples, could be the ability of a node to become inserted between two neghboring nodes (while pushing them towards opposite directions). This could enable parallel mergings of two lines of length $n/2$ into a line of length $n$ in a single step (an, thus, for example, transforming a square to a line in polylogarithmic time). Another, is the capability of rotating whole lines of nodes (like rotating arms, see, e.g., \cite{WCG13}).

There are also some promising specific technical questions. We do not yet know what is the complexity of {\sc RotC-Transformability}. The fact that a 6-seed is capable of transfering pairs of nodes to desired positions, suggests that shapes having such a seed in their exterior or being capable of self-extracting such a seed, will possibly be able to transform to each other. Even if this turns out to be true, it is totally unclear whether transformations involving at least one of the rest of the shapes are feasible. 

Moreover, we didn't study the problem of computing or approximating the optimum transformation. It seems that the problem is computationally hard. A possible approach to prove $\rem{NP}$-hardness would be by proving $\rem{NP}$-hardness of {\sc Rectilinear Graphic TSP} (could be via a reduction from {\sc Rectilinear Steiner Tree} or {\sc Rectilinear TSP}, which are both known to be $\rem{NP}$-complete \cite{GGJ76}) and then giving a reduction from that problem to the problem of a 2-seed exploring a set of locations on the grid.
  
Finally, regarding the distributed transformations, there are various interesting variations of the model considered here, that would make sense. One of them is to assume nodes that are oblivious w.r.t. their orientation.






\end{document}